\documentclass[a4paper,11pt,final]{scrartcl}
\usepackage{ifluatex}
\usepackage[utf8]{inputenc}
\usepackage[T1]{fontenc}
\usepackage{ifpdf}
\usepackage{latexsym}
\usepackage{exscale}
\usepackage{amsfonts}
\usepackage{eucal}
\usepackage{cmmib57}
\usepackage[centertags]{amsmath}
\usepackage{amssymb}
\usepackage{amsbsy}
\usepackage{amstext}
\usepackage{amsthm,thmtools}
\usepackage{abstract}
\usepackage{xcolor}
\usepackage{framed}
\usepackage{etoolbox}
\usepackage[ruled,vlined,algosection]{algorithm2e}
\usepackage{delarray}
\usepackage[automark]{scrlayer-scrpage}
\usepackage{times}
\usepackage[toc,page]{appendix}
\usepackage{textcomp}
\usepackage{mflogo}
\usepackage{here}
\usepackage[round]{natbib}
\usepackage{here}
\usepackage{longtable}
\usepackage{threeparttablex}
\usepackage{multirow}
\usepackage{multicol}
\usepackage{caption}
\usepackage{lscape}
\usepackage{rotating,capt-of}
\usepackage[margin=15mm]{geometry}
\usepackage{tikz}
\usetikzlibrary{shapes}

\newcounter{treeline}

\ifvtexpdf\pdftrue\fi
\ifpdf
\usepackage{graphicx}
\usepackage{nameref}
\usepackage[pdftex,colorlinks]{hyperref}
\ifluatex
\else
\fi 
\hypersetup{
           breaklinks=true,
           pdfpagemode=UseOutlines,
           pdfstartview=FitH,
           colorlinks=true,
           linkcolor=blue,
           citecolor=blue,
           bookmarksopen=true,
           pdftitle={Polynomial-Time Algorithms for Computing the Nucleolus: An Assessment},
           pdfauthor={Holger Ingmar Meinhardt},
           pdfsubject={Polynomial-Time Algorithms of the Pre-Kernel, Nucleolus, TU Games},
           pdfkeywords={Fairness, Pre-Kernel, Pre-Nucleolus, Nucleolus, TU Games},
}

\else
\usepackage{epsfig}
\usepackage{nameref}
\usepackage[ps2pdf]{hyperref}

\hypersetup{
           breaklinks=true,
           pdfpagemode=UseOutlines,
           pdfstartview=FitH,
           colorlinks=true,
           linkcolor=blue,
           citecolor=blue,
           bookmarksopen=true,
           pdftitle={Polynomial-Time Algorithms for Computing the Nucleolus: An Assessment},
           pdfauthor={Holger Ingmar Meinhardt},
           pdfsubject={Polynomial-Time Algorithms of the Pre-Kernel, Nucleolus, TU Games},
           pdfkeywords={Fairness, Pre-Kernel, Pre-Nucleolus, Nucleolus, TU Games},
}
\fi

\addtocounter{MaxMatrixCols}{20}
\let\arrowvert\vert
\let\Arrowvert\Vert
\makeatletter
\renewcommand{\section}{\scr@startsection
  {section}%
  {1}%
  {0em}%
  {-\baselineskip}%
  {0.5\baselineskip}%
  {\centering\normalfont\Large\scshape\mdseries}}%
\makeatother

\makeatletter
\renewcommand{\subsection}{\scr@startsection
  {subsection}%
  {2}%
  {0em}%
  {-\baselineskip}%
  {0.5\baselineskip}%
  {\normalfont\large\scshape\mdseries}}%
\makeatother

\makeatletter
\renewcommand*\env@matrix[1][c]{\hskip -\arraycolsep
  \let\@ifnextchar\new@ifnextchar
  \array{*\c@MaxMatrixCols #1}}
\makeatother

\makeatletter

\newenvironment{theopargself*}
    {\def\@spopargbegintheorem##1##2##3##4##5{\trivlist
         \item[\hskip\labelsep{##4##1\ ##2}]{\hspace*{-\labelsep}##4##3\@thmcounterend}##5}
     \def\@Opargbegintheorem##1##2##3##4{##4\trivlist
         \item[\hskip\labelsep{##3##1}]{\hspace*{-\labelsep}##3##2\@thmcounterend}}}{}
\def \@floatboxreset {%
        \reset@font
        \small
        \@setnobreak
        \@setminipage
}
\def\figure{\@float{figure}}
\@namedef{figure*}{\@dblfloat{figure}}
\def\table{\@float{table}}
\@namedef{table*}{\@dblfloat{table}}
\def\fps@figure{htbp}
\def\fps@table{htbp}

\makeatother

\renewcommand{\thetable}{\thesection.\arabic{table}}

\makeatletter
\@addtoreset{equation}{section}
\makeatother
\makeatletter
\@addtoreset{table}{section}
\makeatother

\theoremstyle{plain}
\newtheorem{theorem}{Theorem}[section]
\newtheorem{proposition}{Proposition}[section]
\newtheorem{corollary}{Corollary}[section]
\newtheorem{lemma}{Lemma}[section]
\newtheorem{definition}{Definition}[section]

\newtheoremstyle{break}
  {9pt}
  {9pt}
  {\itshape}
  {}
  {\bfseries}
  {.}
  {\newline}
  {}

\newtheoremstyle{break1}
  {9pt}
  {9pt}
  {\rmfamily}
  {}
  {\scshape}
  {.}
  {\newline}
  {}

\theoremstyle{break}

\newtheoremstyle{note}
  {3pt}
  {3pt}
  {}
  {}
  {\itshape}
  {:}
  {.5em}
  {\newline}  
  {}

\theoremstyle{note}

\theoremstyle{definition}
\newtheorem{example}{Example}[section]

\theoremstyle{break1}

\begin{document}
\bibliographystyle{plainnat} 
\pdfbookmark[0]{Polynomial-Time Algorithms for Computing the Nucleolus: An Assessment}{tit}
\title{{Polynomial-Time Algorithms for Computing the Nucleolus: An Assessment} }
\author{{\bfseries Holger I. MEINHARDT}~\thanks{The author acknowledges support by the state of Baden-Württemberg through bwHPC. In particular, the kind and excellent technical support supplied by Hany Ibrahim, and Peter Weisbrod is acknowledged. Of course, the usual disclaimer applies.}
\thanks{Holger I. Meinhardt, Institute of Operations Research, Karlsruhe Institute of Technology (KIT), Englerstr. 11, Building: 11.40, D-76128 Karlsruhe. E-mail: \href{mailto:holger.meinhardt@partner.kit.edu}{holger.meinhardt@partner.kit.edu}, ORCID: \href{https://orcid.org/0000-0002-8990-4190}{https://orcid.org/0000-0002-8990-4190}.} 
}
\maketitle

\begin{abstract}
Recently, \citet{maggiorano:2025} claimed that they have developed a strongly polynomial-time combinatorial algorithm for the nucleolus in convex games that is based on the reduced game approach and submodular function minimization method. Thereby, avoiding the ellipsoid method with its negative side effects in numerical computation completely. However, we shall argue that this is a fallacy based on an incorrect application of the Davis/Maschler reduced game property (RGP). Ignoring the fact that despite the pre-nucleolus, other solutions like the core, pre-kernel, and semi-reactive pre-bargaining set possess this property as well. This causes a severe selection issue, leading to the failure to compute the nucleolus of convex games using the reduced games approach. In order to assess this finding in its context, the ellipsoid method of~\citet{faiglekernkuip:01} and the Fenchel-Moreau conjugation-based approach from convex analysis of~\citet{mei:13} to compute a pre-kernel element were resumed. In the latter case, it was exploited that for TU games with a single-valued pre-kernel, both solution concepts coincide. Implying that one has computed the pre-nucleolus if one has found the sole pre-kernel element of the game. Though it is a specialized and highly optimized algorithm for the pre-kernel, it assures runtime complexity of $O(n^{3})$ for computing the pre-nucleolus whenever the pre-kernel is a single point and the~\citet{faiglekernkuip:01} assumption is satisfied, which indicates a polynomial-time algorithm for this class of games.

\vspace{3em}


\end{abstract}

\pagestyle{scrheadings}  \ihead{\empty} \chead{Polynomial-Time Algorithms for Computing the Nucleolus} \ohead{\empty}

\section{Introduction}
\label{subsec:polytimeAlgo}

Through the work of~\citet{kop:67,kohl:72,ow:74}, several universal linear programming approaches have been developed to compute the pre-nucleolus of a TU game. In this context,~\citet[pp. 146-148]{pel_sud:07} have established that at most a sequence of $(n-1)$ LPs must be iteratively solved to successfully find the pre-nucleolus of the TU game, whereas each of these LPs has exponentially many constraints. Thus, computing the nucleolus through this method requires solving a sequence of linear programs, which has exponential complexity in the number of players $n$ due to the number of $2^{n}$ possible coalitions. Instead of directly approaching this large number of constraints, it is wiser to use an indirect method in the computation of the nucleolus to drastically reduce the order of the constraint set by reformulating the initial problem. In the sequel, we shall present and assess some polynomial-time algorithms that have the claim for indirectly computing the nucleolus for certain game classes.

\citet{faiglekernkuip:01} provided a polynomial-time algorithm for computing the nucleolus for the class of games having a sole intersection point of the pre-kernel with the core. Their proposed computation procedure relies on the ellipsoid method to solve the associated least-cores and on Maschler's scheme for approximating the (pre-)kernel. In this respect, sufficiently many steps of Maschler's approximation scheme must be performed to find suitable hyperplanes for applying the ellipsoid algorithm. Note that their polynomial-time result is critical based on the assumption that the minimum excess of any allocation can be efficiently computed, i.e., by a polynomial-time method. This is owed to the fact that the computation of the nucleolus is, in general, {\bfseries NP}-hard, as it has been worked out by~\citet{faiglekernkuip:98}. Though these authors deal with the special class of minimal-cost spanning-tree games. 

By this methodological approach, the general {\bfseries NP}-hardness of the nucleolus computation is deduced from a representative problem. Establishing {\bfseries NP}-hardness from this class of games to the entire class of TU games, therefore, imposes no loss of generality. To see this, we have to recognize that the problem of computing the nucleolus of a minimum-cost spanning-tree game has an equivalent structure to the problem of computing the nucleolus for a general TU game. Hence, the problem of computing the nucleolus of a minimum-cost spanning-tree game is representative of solving the nucleolus of a general TU game. Now,~\citeauthor{faiglekernkuip:98} prove that the representative problem for computing the nucleolus has a reduction to a minimum cover game. It is well-known that finding a minimum cover is {\bfseries NP}-hard. Finally, they show that solving the representative problem for the nucleolus is at least as difficult to solve as a minimum cover game, which is from the complexity class {\bfseries NP}. This argument is sufficient to conclude that computing the nucleolus is, in general, {\bfseries NP}-hard.    

However, to avoid the well-known problem of unbounded complexity in a real number model related to the ellipsoid algorithm,~\citet{faiglekernkuip:01} imposed a rational number model instead. The unbounded complexity in a real number model implies numerical instability issues with floating-point arithmetic, causing the ellipsoid algorithm to fail or converge to an incorrect solution. These negative side effects crucially limit the practical relevance of computing the nucleolus through the ellipsoid method. Nevertheless, their finding must be understood as a major theoretical breakthrough in proving polynomial-time solvability of the nucleolus for a relatively large class of games, and not just for special game classes such as assignment games (cf.~\citet{solrag:94}), matching games (cf.~\citet{wkern:03}), or standard tree games (cf.~\citet{megiddo:78,granotmasch:96}).

Hence, there is still a strong need for a universally valid practical method in the computation of the nucleolus by a polynomial-time algorithm that overcomes all these negative side effects. \citet{maggiorano:2025} recently claimed that they were able to close this gap, at least in the domain of convex games. In particular, these authors claimed to provide a first strongly polynomial-time algorithm for computing the nucleolus of convex games while leveraging the Davis/Maschler reduced game approach in combination with submodular function minimization. In this context, the submodular function minimization is applied in order to catch the nucleolus by the least-core of the game. This avoids the ellipsoid method with its negative side effects in numerical computation completely for solving the linear programming problems related to the least-core of the default game and the emerging Davis/Maschler reduced games. However, we shall argue that this is a fallacy based on an incorrect application of the Davis/Maschler reduced game property (RGP). Ignoring the fact that despite the pre-nucleolus, also other solutions like the core, pre-kernel (cf.~\citet{Pel:86b}), and semi-reactive pre-bargaining set (cf.~\citet{SudPot:01}) possess this property, whereas the list of mentionable solution concepts is not complete. This causes a severe selection issue, leading to the failure to compute the nucleolus of convex games using the reduced games approach. 

A feature that both methods from above have in common is that they leverage the least-core as the catcher of the nucleolus. At the cost that the problem must be overloaded with too much structure, that is not necessary for efficiently computing the nucleolus of the game. A structure that is imposed by the original characterization of the nucleolus, which is inherently {\bfseries NP}-hard. It should be clear from this that without a substantial simplification of the nucleolus characterization, it will not be possible to develop a universal polynomial-time algorithm for calculating the nucleolus. The work of~\citet{kido:04,kido:05,kido:08} has pointed out an interesting direction in which the journey towards a polynomial-time solvability of the nucleolus should go without referring to any form of nucleolus catcher. Unfortunately, it is a widely unrecognized work in the literature despite its theoretical elegance and power as a tool for computing the nucleolus. 

The underlying idea was only taken up again by~\citet{mei:13}, not for the calculation of the nucleolus, but for the calculation of the pre-kernel as a solution concept used. In this framework, a Fenchel-Moreau conjugation-based approach from convex analysis was applied to compute a pre-kernel element. While the pre-nucleolus is always a single point within the pre-kernel, both solution concepts are just for some special game classes identical, like 3-person games, convex games, veto-rich games, or almost-convex games. Thus, if the pre-kernel is identical to the pre-nucleolus, one has computed the pre-nucleolus if one has found the sole pre-kernel element of the game. This fact is exploited when applying the Fenchel-Moreau approach to calculate the pre-nucleolus of a game with a single-valued pre-kernel.

This procedure offers several theoretical and computational advantages over general-purpose methods like the ellipsoid method by reformulating the problem in a highly optimized way. In particular, no nucleolus catcher was needed for imposing its successful computation. Moreover, due to the fact that the initial problem is formulated as an overdetermined system of equations, an equivalent optimization problem can be constructed to iteratively solve a sequence of linear equations, which does scale extremely well. Hence, it is a practical method for computing this single point. Also, polynomial-time solvability is guaranteed when a similar assumption as under~\citet{faiglekernkuip:01} is imposed; then the algorithm reduces to iteratively solving a system of linear equations with a worst-case runtime complexity of $O(n^3)$. In addition, through the induced iterative process of orthogonal projections, a fast convergence is guaranteed. Then, this method converges in a finite, small number of steps; typically just $(n+1)$ iteration steps are needed to complete the computation. Thus, applying this approach in the nucleolus computation is more efficient because of its specialized, elegant, and highly optimized algorithm for the pre-kernel.

These findings and their implications were completely ignored in the search for a polynomial-time algorithm for the nucleolus. Instead, the authors resorted to outdated ideas like nucleolus catchers in a new formal framework, which require cumbersome constructions to ensure successful nucleolus computation. This is absolutely counterproductive in conceiving a truly efficient method for computing the nucleolus in general. By contrast, the Fenchel-Moreau-based approach illustrates the potential of a substantial simplification of the characterization of the nucleolus for the class of games with a single-valued pre-kernel in order to conceive a universal polynomial-time algorithm for the nucleolus.

The remainder of the paper is organized as follows: In the Section~\ref{sec:prel} we introduce some basic notations and definitions to investigate the coincidence of the pre-kernel with the pre-nucleolus. Then Section~\ref{sec:dprk} provides the concept of the indirect function and gives a dual pre-kernel representation in terms of solution sets. In the next step, the notion of lexicographically smallest most effective coalitions is introduced in order to identify payoff equivalence classes on the domain of the objective function from which a pre-kernel element can be determined. Moreover, to keep the forthcoming assessment of a sample of algorithms for computing the nucleolus in polynomial runtime complexity as self-explanatory as possible, we sketch in Section~\ref{sec:effsolvprob} some characterizations of efficient solvability of a problem. For then, to turn to the ellipsoid method of computing the nucleolus for some classes of games in polynomial-time in Section~\ref{sec:methellip}. After that we provide a sketch of the ellipsoid method (Section~\ref{sec:skellip}) and of Maschler's scheme for approximating the (pre-)kernel of a game (Section~\ref{sec:approxkr}). Having clarified their main features, we are in a position to turn our attention to the runtime complexity of the ellipsoid method for the rational (Section~\ref{sec:elliprat}) and real number model (Section~\ref{sec:ellipreal}) before looking at its theoretical and practical limitations in Section~\ref{sec:unbcomp}. Then, we look at two recent methods of computing the nucleolus that avoid the ellipsoid method. The first one (Section~\ref{sec:rgpsfm}) is based on a submodular function minimization to enclose the nucleolus again within the least-core as its nucleolus catcher to impose the Davis/Maschler reduced-game property of trying to single out the nucleolus for convex games. By contrast, the second one (Section~\ref{sec:mthdindfuc}) is based on our Fenchel-Moreau conjugation method for finding a pre-kernel element, for which we establish its usefulness for certain game classes in the computation of the nucleolus by means of polynomial runtime complexity. Finally, we focus in Section~\ref{sec:repprk} on the replication of the nucleolus in the sense of~\citet{mei:23} for a class of related games to demonstrate the flexibility of the Fenchel-Moreau-based algorithm. This is underpinned by providing a manual step-by-step procedure for computing the nucleolus for a convex game borrowed from the literature. In the Section~\ref{sec:concrem}, our assessment of a sample of polynomial-time algorithms for the nucleolus is closed by some concluding remarks.

\section{Some Preliminaries}
\label{sec:prel}
A cooperative game with transferable utility is a pair $\langle N,v \rangle $, where $N$ is the non-empty finite player set $N := \{1,2, \ldots, n\}$, and $v$ is the characteristic function $v: 2^{N} \rightarrow \mathbb{R}$ with $v(\emptyset):=0$. A player $i$ is an element of $N$, and a coalition $S$ is an element of the power set of $2^{N}$. The real number $v(S) \in \mathbb{R}$ is called the value or worth of a coalition $S \in 2^{N}$. Let $S$ be a coalition, the number of members in $S$ will be denoted by $s:=|S|$. We assume throughout that $v(N) > 0$ and $n \ge 2$ is valid. In addition, we identify a cooperative game by the vector $v := (v(S))_{S \subseteq N} \in \mathcal{G}^{n} = \mathbb{R}^{2^{n}}$, if no confusion can arise. Finally, the relevant game space for our investigation is defined by $\mathcal{G}(N) := \{v \in \mathcal{G}^{n}\,\arrowvert\, v(\emptyset) = 0 \land v(N) > 0\}$.  

If $\mathbf{x} \in \mathbb{R}^{n}$, we apply $x(S) := \sum_{k \in S}\, x_{k}$ for every $S \in 2^{N}$ with $x(\emptyset):=0$. The set of vectors $\mathbf{x} \in \mathbb{R}^{n}$ which satisfies the efficiency principle $v(N) = x(N)$ is called the {\bfseries pre-imputation set} and it is defined by 
\begin{equation} 
  \label{eq:pre-imp}
  \mathcal{I}^{*}(N,v):= \left\{\mathbf{x} \in \mathbb{R}^{n} \;\arrowvert\, x(N) = v(N) \right\}, 
\end{equation} 
or more concisely as  $\mathcal{I}^{*}(v)$, where an element $\mathbf{x} \in \mathcal{I}^{*}(v)$ is called a pre-imputation. The set of pre-imputations which satisfies in addition the {\bfseries individual rationality property} $x_{k} \ge v(\{k\})$ for all $k \in N$ is called the {\bfseries imputation set} $\mathcal{I}(N,v)$.

A vector that results from a vector $\mathbf{x}$ by a {\bfseries transfer}\index{transfer} of size $\delta \ge 0$ between a pair of players $i, j \in N, i \neq j$, is referred to as $\mathbf{x}^{\;i,j,\delta} = (x^{\;i,j,\delta}_{k})_{k \in N}$, which is given by
\begin{equation}
 \label{eq:sidepyA}
  \mathbf{x}^{\;i,j,\delta}_{N\backslash\{i,j\}} = \mathbf{x}_{N\backslash\{\;i,j\}},\; x^{i,j,\delta}_{i} = x_{i} - \delta\quad\text{and}\quad x^{\;i,j,\delta}_{j} = x_{j} + \delta.
\end{equation}
A {\bfseries side-payment}\index{side-payment} for the players in $N$ is a vector $\mathbf{z} \in \mathbb{R}^{n}$ such that $z(N) = 0$ holds. 

A {\bfseries solution concept}, denoted as $\sigma$, on a non-empty set $\mathcal{G}$ of games is a correspondence on $\mathcal{G}$ that assigns to any game $v \in \mathcal{G}$ a subset $\sigma(N,v)$ of $\mathcal{I}^{*}(N,v)$\index{solution concept!set}. This set can be empty or just be single-valued, in the latter case, the solution $\sigma$ is a function and is simply called a value.

The {\bfseries core} of a game $\langle N, v \rangle$ is a set-valued solution that is constituted by the imputations satisfying besides the individual rationality property as well as the coalitional rationality property, i.e. the core of a game $v \in \mathcal{G}^{n}$ is given by 
\begin{equation}
 \label{eq:core_solA}
  \mathcal{C}(N,v):= \left\{\mathbf{x} \in \mathcal{I}(N,v)\,\arrowvert\, x(N) = v(N)\;\text{and}\; x(S) \ge v(S) \;\forall\;  S \subset N\right\}.
\end{equation}
The core of a $n$-person game may be empty. Whenever it is non-empty we have some incentive for mutual cooperation in the grand coalition.

Given a vector $\mathbf{x} \in \mathcal{I}^{*}(v)$, we define the {\bfseries excess} of coalition $S$ with respect to the pre-imputation $\mathbf{x}$ in the game $\langle N,v \rangle $ by 
\begin{equation} 
  \label{eq:exc} 
  e^{v}(S,\mathbf{x}):= v(S) - x(S). 
\end{equation} 

Take a game $v \in \mathcal{G}^{n}$. For any pair of players $i,j \in N, i\neq j$, the {\bfseries maximum surplus} of player $i$ over player $j$ with respect to any pre-imputation $\mathbf{x} \in \mathcal{I}^{*}(v)$ is given by the maximum excess at $\mathbf{x}$ over the set of coalitions containing player $i$ but not player $j$, thus\begin{equation} 
  \label{eq:maxexc} 
  s_{ij}(\mathbf{x},v):= \max_{S \in \mathcal{G}_{ij}} e^{v}(S,\mathbf{x}) \qquad\text{where}\;  \mathcal{G}_{ij}:= \{S \;\arrowvert\; i \in S\; \text{and}\; j \notin S \}. 
\end{equation} 
The set of all pre-imputations $\mathbf{x} \in \mathcal{I}^{*}(v)$ that balances the maximum surpluses for each distinct pair of players $i,j \in N, i\neq j$ is called the~\hypertarget{hyp:prk}{{\bfseries pre-kernel}} of the game $v$, and is defined by 
  \begin{equation} 
    \label{eq:prekA} 
    \mathcal{PK}(v) := \left\{ \mathbf{x} \in \mathcal{I}^{*}(v)\; \arrowvert\;  s_{ij}(\mathbf{x},v) = s_{ji}(\mathbf{x},v) \quad\text{for all}\; i,j \in N, i\neq j \right\}. 
  \end{equation} 
The pre-kernel has the advantage of addressing a stylized bargaining process, in which the figure of argumentation is a {\bfseries pairwise equilibrium procedure} of claims while relying on best arguments, that is, the coalitions that will best support the claim. The pre-kernel solution characterizes all those imputations in which all pairs of players $i,j \in N, i\neq j$ are in equilibrium with respect to their claims.

In order to define the pre-nucleolus of a game $v \in \mathcal{G}^{n}$, take any $\mathbf{x} \in \mathbb{R}^{n}$ to define a $2^{n}$-tuple vector $\theta(\mathbf{x})$ whose components are the excesses $e^{v}(S,\mathbf{x})$ of the $2^{n}$ coalitions $S \subseteq N$, arranged in decreasing order, that is,
\begin{equation}
 \label{eq:compl_vec}
  \theta_{i}(\mathbf{x}):=e^{v}(S_{i},\mathbf{x}) \ge e^{v}(S_{j},\mathbf{x}) =:\theta_{j}(\mathbf{x}) \qquad\text{if}\qquad 1 \le i \le j \le 2^{n}.
\end{equation}
Ordering the so-called complaint or dissatisfaction vectors $\theta(\mathbf{x})$ for all $\mathbf{x} \in \mathbb{R}^{n}$ by the lexicographic order  $\le_{L}$ on $\mathbb{R}^{n}$, we shall write
\begin{equation}
 \theta(\mathbf{x}) <_{L} \theta(\mathbf{y}) \qquad\text{if}\;\exists\;\text{an integer}\; 1 \le k \le 2^{n},
\end{equation}
such that $\theta_{i}(\mathbf{x}) = \theta_{i}(\mathbf{y})$ for $1 \le i < k$ and $\theta_{k}(\mathbf{x}) < \theta_{k}(\mathbf{y})$. Furthermore, we write $\theta(\mathbf{x}) \le_{L} \theta(\mathbf{y})$ if either $\theta(\mathbf{x}) <_{L} \theta(\mathbf{y})$ or $\theta(\mathbf{x}) = \theta(\mathbf{y})$. Now the pre-nucleolus $\mathcal{PN}(N,v)$ over the pre-imputations set $\mathcal{I}^{*}(v)$ is defined by 
\begin{equation}
 \label{eq:prn_sol}
  \mathcal{PN}(N,v) = \left\{\mathbf{x} \in \mathcal{I}^{*}(N,v)\; \arrowvert\; \theta(\mathbf{x}) \le_{L} \theta(\mathbf{y}) \;\forall\; \mathbf{y} \in \mathcal{I}^{*}(N,v) \right\}.
\end{equation}
The {\bfseries pre-nucleolus} of any game $v \in \mathcal{G}^{n}$ is non-empty as well as unique, and it is referred to as $\nu(v)$ if the game context is clear from the contents or $\nu(N,v)$ otherwise. 

Moreover, both solutions can be uniquely characterized by a set of axioms. In order to formalize such an axiomatization, let $\langle\, N, v\,\rangle \in \mathcal{G}$ be a game s.t.~$\emptyset \neq S \subseteq N$ and let $\vec{x} \in \mathcal{I}^{*}(N,v)$. The {\bfseries Davis/Maschler reduced game} w.r.t.~$S$ and $\vec{x}$ is the game $\langle\, S, v_{S,\,\vec{x}}\,\rangle$ as given by
\begin{equation}
 \label{eq:rdg_modic}
  v_{S,\,\vec{x}}(T):= 
  \begin{cases}
    0 & \text{if}\; T = \emptyset \\
    v(N)-x(N\backslash S) & \text{if}\; T = S \\
    \max_{Q \subseteq N\backslash S}\,\left(v(T \cup Q) - x(Q) \right) & \text{otherwise.}
  \end{cases}
\end{equation}
This game type has been introduced by~\citet{davis:65} to study the kernel.

\begin{definition}[DM-RGP]
 \label{def:hm_rgp0}
A solution $\sigma$ on $\mathcal{G}$ satisfies the Davis/Maschler reduced game property (RGP), if for $\langle N, v \rangle  \in \mathcal{G}, \emptyset \neq S \subseteq N$ and $\mathbf{x} \in \sigma(N,v)$, then $\langle\, S, v_{S,\,\mathbf{x}}\,\rangle  \in \mathcal{G}$ and $\mathbf{x}_{S} \in \sigma(S,v_{S,\,\mathbf{x}})$.  
\end{definition}

Let $\sigma$ be a solution concept on the set $\mathcal{G}$, and $\mathcal{U}$ the universe of players.   In addition, define the permutation group by $\text{Sym}(N):=\{\vartheta: N \rightarrow N \,\arrowvert \vartheta \;\text{is bijective} \}$ acting on the game space $\mathcal{G}$ by linear transformations. Hence, each bijection $\vartheta \in  \text{Sym}(N)$ corresponds to a linear and invertible transformation of an element of the vector space $\mathcal{G}$ by defining a permuted game $\vartheta \, v(S):=v(\vartheta^{-1}\, S)$ for every $\vartheta \in \text{Sym}(N), v \in \mathcal{G}$ and $S \subseteq N$, whereas $\vartheta\, S := \{\vartheta(i) \,\arrowvert\, i \in S \}$ and  $\vartheta^{-1}\, S := \{i \,\arrowvert\, \vartheta(i) \in S \}$. Hence, the games $\langle\, N, v\,\rangle$ and $\langle\, \vartheta N, \vartheta v\,\rangle$ are equivalent, where we have written the group operations for the sake of convenience as junction.  

  \begin{enumerate}
  \item A solution $\sigma$ on $\mathcal{G}$ satisfying {\bfseries non-emptiness (NE)} if $\sigma(N,v) \neq \emptyset$ for every $\langle\, N, v\,\rangle  \in \mathcal{G}$.
  \item A solution $\sigma$ on $\mathcal{G}$ fulfills the {\bfseries individual rationality property (IRP)}  if for $\langle\, N, v\,\rangle \in \mathcal{G}$ and $\vec{x} \in \sigma(N,v)$, then $x_{k} \ge v(\{k\})$ for all $k \in N$.     
  \item A solution $\sigma$ on $\mathcal{G}$ is {\bfseries single-valued (SIVA)} if $\arrowvert\,\sigma(N,v)\, \arrowvert = 1$ for every $\langle\, N, v\,\rangle  \in \mathcal{G}$.
  \item A solution $\sigma$ on $\mathcal{G}$ satisfies {\bfseries efficiency (EFF)} if for all $\langle\, N, v\,\rangle  \in \mathcal{G}$, then $\sum_{k \in N}\,\sigma_{k}(N,v) = v(N) $.  
  \item A solution $\sigma$ on $\mathcal{G}$ satisfies the {\bfseries equal treatment property (ETP)}, if $\langle\, N, v\,\rangle \in \mathcal{G}$, $\vec{x} \in \sigma(N,v)$ and if $k,l \in N$ s.t. $k \sim_{v} l$, then $x_{k} = x_{l}$.
  \item A solution $\sigma$ on $\mathcal{G}$ satisfies {\bfseries anonymity (AN)} if for $\langle\, N, v\,\rangle  \in \mathcal{G}$, for a bijection $\vartheta \in  \text{Sym}(N)$ and for $\langle\, \vartheta N, \vartheta v\,\rangle \in \mathcal{G}$ implying $\sigma(\vartheta N, \vartheta v) = \vartheta(\sigma(N,v))$.
  \item A solution $\sigma$ on $\mathcal{G}$ satisfies {\bfseries superadditivity (SUPA)} if for all $\langle\, N, v_{1}\,\rangle, \langle\, N, v_{2}\,\rangle, \langle\, N, v_{1}+v_{2}\,\rangle \in \mathcal{G}$, then $\sigma(v_{1}+v_{2})  \supseteq \sigma(v_{1}) + \sigma(v_{2})$ is satisfied.  
  \item A solution $\sigma$ on $\mathcal{G}$ fulfills the {\bfseries Covariance under Strategic Equivalence (COV)} property if for $\langle\, N, v_{1}\,\rangle, \langle\, N, v_{2}\,\rangle \in \mathcal{G}$, with $v_{2} = t \cdot v_{1} + \mathfrak{m} $ for some $t \in \mathbb{R}\backslash\{0\}, \mathfrak{m} \in \mathbb{R}^{2^{n}}$, then $\sigma(N,v_{2}) = t \cdot \sigma(N,v_{1}) + \mathbf{m}$, whereas $\mathbf{m} \in \mathbb{R}^{n}$ and $\mathfrak{m}$ is the vector of measures obtained from $\mathbf{m}$.
 \item A solution $\sigma$ on $\mathcal{G}$ possesses the {\bfseries converse reduced game property (CRG)} property if for $\langle\, N, v\,\rangle  \in \mathcal{G}$ with $|N| \ge 2$, $\vec{x} \in \mathcal{I}^{*}(N,v)$, $\langle\, S, v_{S,\,\vec{x}}\,\rangle  \in \mathcal{G}$ and $\vec{x}_{S} \in \sigma(S,v_{S,\,\vec{x}})$ for every $S \in \{T \subseteq N\,\arrowvert\, |T|=2 \}$, then $\vec{x} \in \sigma(N,v)$.
  \end{enumerate}

\begin{theorem}[\citet{Sob:75}]
  \label{thm:sob}
If $\mathcal{U}$ is an infinite player set, then there exists a unique solution $\sigma$ on $\mathcal{G}_{\mathcal{U}}$ satisfying single-valuedness (SIVA), anonymity (AN), covariance under strategic equivalence (COV), and reduced game property (RGP), which is the pre-nucleolus. 
\end{theorem}

By the above theorem, we recognize that the {\bfseries RGP} axiom is just one axiom among others to characterize the pre-nucleolus. It does not characterize the pre-nucleolus alone. Moreover, recall that the pre-kernel is characterized by the following set of logically independent axioms: {\bfseries NE, EFF, COV, ETP, RGP}, and {\bfseries CRGP} (cf.~\citet[Section 5.4]{pel_sud:07}). Hence, we observe that even the pre-kernel is characterized by the {\bfseries RGP} axiom. By contrast, the core is characterized on the set of balanced games by the following three logically independent axioms: {\bfseries IRP, SUPA}, and {\bfseries RGP} (cf.~\citet[Section 5]{Pel:86b}). Again, we observe that the {\bfseries RGP} axiom can also be attributed to the core of balanced games. Therefore, the {\bfseries RGP} is not an axiom that exclusively characterizes the pre-nucleolus.  

In addition, we want to discuss some important game properties. A game $v \in \mathcal{G}^{n}$ is said to be {\bfseries monotonic} if 
\begin{equation}
 \label{eq:mono_propA}
  v(S) \le v(T) \qquad\forall \emptyset \neq S \subseteq T.
\end{equation}
Thus, whenever a game is monotonic, a coalition $T$ can guarantee to its member a value at least as high as  any sub-coalition $S$ can do. This subclass of games is referred to as $\mathcal{MN}^{n}$. A game $v \in \mathcal{G}^{n}$ satisfying the condition
\begin{equation}
 \label{eq:sadd_propA}
  v(S) + v(T) \le v(S \cup T) \qquad\forall S,T \subseteq N,\;\text{with}\; S \cap T = \emptyset,
\end{equation}
is called {\bfseries superadditive}. This means, that two disjoint coalitions have some incentive to join into a mutual coalition. This can be regarded as an incentive of merging economic activities into larger units. We denote this subclass of games by $\mathcal{SA}^{n}$. However, if a game $v \in \mathcal{G}^{n}$ satisfies
\begin{equation}
 \label{eq:conv_propA}
  v(S) + v(T) \le v(S \cup T) + v(S \cap T)\qquad\forall S,T \subseteq N,
\end{equation}
or equivalently
\begin{equation}
 \label{eq:conv_propA2}
  v(S \cup \{i\}) - v(S)  \le v(S \cup \{i,j\} ) - v(S \cup \{j\})  \qquad\text{if} \; S \subseteq  N \backslash \{i,j\},
\end{equation}
then it is called {\bfseries convex}. In this case, we will observe a strong incentive for a mutual cooperation in the grand coalition, due to its achievable over proportionate surpluses while increasing the scale of cooperation. This subclass of games has been introduced by~\citet{Shapley:71}, and we denote it by $\mathcal{CV}$. Convex games having a non-empty core and the Shapley value is the center of gravity of the extreme point of the core (cf.~\citet{Shapley:71}), that is, a convex combination of the vectors of marginal contributions, which are core imputations for convex games. It should be evident that $\mathcal{CV}^{n} \subset \mathcal{SA}^{n}$ is satisfied.

\section{A Dual Pre-Kernel Representation}
\label{sec:dprk}

Now, we present an alternative approach to characterizing the pre-kernel in terms of the so-called indirect function, introduced by~\citet{mart:96}. The indirect function of a characteristic function game -- which is a Fenchel-Moreau generalized conjugation -- provides the same information as the $n$-person cooperative game with transferable utility under consideration. In this section, we review some crucial results extensively studied in~\citet[Chap.~5 \&~6]{mei:13} as the building blocks to investigate the computational aspect of the pre-kernel.

The {\bfseries convex conjugate} or {\bfseries Fenchel transform} $f^{*}: \mathbb{R}^{n} \to \overline{\mathbb{R}}$ (where $\overline{\mathbb{R}} := \mathbb{R} \cup \{ \pm\;\infty\}$) of a convex function $f: \mathbb{R}^{n} \to \overline{\mathbb{R}}$ (cf.~\citet[Section 12]{Rocka:70}) is defined by 
\begin{equation*} 
  f^{*}(\mathbf{x}^{\,*}) = \sup_{\mathbf{x} \in \mathbb{R}^{n}} \{\langle\; \mathbf{x}^{\,*}, \mathbf{x} \;\rangle - f(\mathbf{x})\} \qquad \forall \mathbf{x}^{\,*} \in \mathbb{R}^{n}.
\end{equation*} 
Observe that the Fenchel transform $f^{*}$ is the point-wise supremum of affine functions $p(\mathbf{x}^{\,*}) = \langle\; \mathbf{x}, \mathbf{x}^{\,*} \;\rangle - \mu$ such that $(\mathbf{x},\mu) \in (\CMcal{C} \times \mathbb{R}) \subseteq (\mathbb{R}^{n} \times \mathbb{R})$, whereas $\CMcal{C}$ is a convex set. Thus, the Fenchel transform $f^{*}$ is again a convex function.

We can generalize the definition of a Fenchel transform (cf.~\citet{mart:96}) by introducing a fixed non-empty subset $\CMcal{K}$ of $\mathbb{R}^{n}$, then the conjugate of a function $f: \CMcal{K} \to \overline{\mathbb{R}}$ is $f^{c}: \mathbb{R}^{n} \to \overline{\mathbb{R}}$, given by
\begin{equation*}
  f^{c}(\mathbf{x}^{\,*}) = \sup_{\mathbf{x} \in \CMcal{K}} \{\langle\; \mathbf{x}^{\,*}, \mathbf{x} \;\rangle - f(\mathbf{x})\} \qquad \forall \mathbf{x}^{\,*} \in \mathbb{R}^{n},
\end{equation*}
which is also known as the {\bfseries Fenchel-Moreau conjugation}.

A vector $\mathbf{x}^{\,*}$ is said to be a subgradient of a convex function $f$ at a point $\mathbf{x}$, if
\begin{equation*}
  f(\mathbf{z}) \ge f(\mathbf{x}) + \langle\; \mathbf{x}^{\,*}, \mathbf{z} - \mathbf{x} \;\rangle \qquad\forall \mathbf{z} \in \mathbb{R}^{n}.
\end{equation*}
The set of all subgradients of $f$ at $\mathbf{x}$ is called the subdifferentiable of $f$ at $\mathbf{x}$ and it is defined by
\begin{equation*}
  \partial f(\mathbf{x}):= \{ \mathbf{x}^{\,*} \in \mathbb{R}^{n}\;\arrowvert\; f(\mathbf{z}) \ge f(\mathbf{x}) + \langle\; \mathbf{x}^{\,*}, \mathbf{z} - \mathbf{x} \;\rangle \quad (\forall \mathbf{z} \in \mathbb{R}^{n})\}.
\end{equation*}
The set of all subgradients $\partial f(\mathbf{x})$ is a closed convex set, which could be empty or may consist of just one point. The multivalued mapping $\partial f: \mathbf{x} \mapsto \partial f(\mathbf{x})$ is called the subdifferential of $f$.

\begin{theorem}[Indirect Function (\citet{mart:96})]
\label{th:mart7}
The indirect function $\pi: \mathbb{R}^{n} \to \mathbb{R}$ of any $n$-person TU game is a non-increasing polyhedral convex function such that 
\begin{itemize}
\item[(i)] $\partial{\pi(\mathbf{x})}{} \cap \{-1, 0\}^{n} \neq \emptyset \qquad\forall \mathbf{x} \in \mathbb{R}^{n}$,
\item[(ii)] $ \{-1,0\}^{n} \subset \bigcup_{\mathbf{x} \in \mathbb{R}^{n}} \partial{\pi(\mathbf{x})}{}$, and
\item[(iii)] $\min_{\mathbf{x} \in \mathbb{R}^{n}}\; \pi(\mathbf{x}) = 0$.
\end{itemize}
Conversely, if $\pi: \mathbb{R}^{n} \to \mathbb{R}$ satisfies $(i)$-$(iii)$ then there exists a unique $n$-person TU game $\langle N,v \rangle$ having $\pi$ as its indirect function, its characteristic function is given by
\begin{equation}
\label{eq:mart15}
v(S) = \min_{\mathbf{x} \in \mathbb{R}^{n}}\bigg\{\pi(\mathbf{x}) + \sum_{k \in S}\; x_{k}\bigg\} \qquad\forall\; S \subseteq N.
\end{equation}
\end{theorem}
\noindent Hence, the indirect function provides the same information as the characteristic function. 

According to the above result, the arising {\bfseries indirect function} $\pi: \mathbb{R}^{n}\to \mathbb{R}_{+}$ is given by
\begin{equation}
  \label{eq:indf}
 \pi(\mathbf{x}) = \max_{S \subseteq N}\, \bigg\{v(S) - \sum_{k \in S}\;x_{k} \bigg\} \qquad\forall\mathbf{x}\in\mathbb{R}^{n},
\end{equation}

This approach allows us to propose a more efficient and simple method to determine a pre-kernel element (cf.~\citet{mei:13}). 

The pre-imputation that comprises the possibility of compensation between a pair of players $i, j \in N, i \neq j$, is denoted as $\mathbf{x}^{\;i,j,\delta} = (x^{\;i,j,\delta}_{k})_{k \in N}\in \mathcal{I}^{*}(N,v)$, with $\delta \ge 0$, which is given by
\begin{equation*}
  \mathbf{x}^{\;i,j,\delta}_{N\backslash\{i,j\}} = \mathbf{x}_{N\backslash\{\;i,j\}},\; x^{i,j,\delta}_{i} = x_{i} - \delta\quad\text{and}\quad x^{\;i,j,\delta}_{j} = x_{j} + \delta.
\end{equation*}

By the next Lemma we shall establish that the indirect function $\pi$ of game $v$ can be related to the maximum surpluses.

\begin{lemma}[Lower Bound on Transfers (\citet{mes:97})]
 \label{lem:mes1}
Let $\langle N,v \rangle$ be an $n$-person cooperative game with side-payments. Let $\pi$ and $s_{ij}$ be the associated indirect function and the maximum surplus of player $i$ against player $j$, respectively. If $\epsilon \in \mathbb{R}$ and $\mathbf{x} \in \mathcal{I}^{\epsilon}(v)$, then the equality:
  \begin{equation*}
   s_{ij}(\mathbf{x},v) = \pi(\mathbf{x}^{\;i,j,\delta}) - \delta
  \end{equation*}
holds for every $i,j \in N,\,i\neq j$, and for every $\delta \ge \delta_{1}(\mathbf{x},v)$, where:
\begin{equation*}
  \delta_{1}(\mathbf{x},v) := \max_{k \in N, S \subset N\backslash\{k\}}\; |v(S \cup \{k\}) - v(S) - x_{k}|. 
\end{equation*}
\end{lemma}
\begin{proof}
  For a proof, consult~\citet{mes:97} or~\citet[Lemma~5.3.1]{mei:13}.
\end{proof}

A characterization of the pre-kernel in terms of the indirect function is due to~\citet{mes:97}. Here, we present this representation for the trivial coalition structure $\mathcal{B}=\{N\}$ only since our algorithm evaluates a pre-kernel element for the grand coalition. Thus, for our purpose, it is enough to concentrate on the trivial coalition structure while restating the following result: 

\begin{proposition}[The Indirect Function Characterizes the Pre-Kernel (\citet{mes:97})]
\label{prop:mese1}
For a TU game with indirect function $\pi$, a pre-imputation $\mathbf{x} \in \mathcal{I}^{*}(N,v)$ is in the pre-kernel of $\langle N,v \rangle$ if, and only if, for every $i,j \in N \; i < j$, and some $\delta \ge \delta_{1}(v,\mathbf{x})$, one receives
\begin{equation*}
 \pi(\mathbf{x}^{\;i,j,\delta}) = \pi(\mathbf{x}^{\;j,i,\delta}).
\end{equation*}
\end{proposition}
\begin{proof}
  For a proof, see~\citet{mes:97} or~\citet[Proposition~5.3.1]{mei:13}.
\end{proof}
\noindent For non-trivial coalition structures, we refer the reader to the book of~\citet{mei:18c}.

\citet{mes:97} was the first who recognized that based on the result of Proposition~\ref{prop:mese1} a pre-kernel element can be derived as a solution of an overdetermined system of non-linear equations. Every overdetermined system can be equivalently expressed as a minimization problem. The set of global minima coalesces with the pre-kernel set. For the trivial coalition structure $\mathcal{B} = \{N\}$ the overdetermined system of non-linear equations is given by 
\begin{equation}
  \label{eq:fij} 
  \begin{cases}
    f_{ij}(\mathbf{x}) = 0 & \forall i,j \in N, i < j\\[.5em]
    f_{0}(\mathbf{x}) = 0 
  \end{cases}
\end{equation}
where, for some $\delta \ge \delta_{1}(\mathbf{x},v)$, 
\begin{equation*}
  f_{ij}(\mathbf{x}) := \pi(\mathbf{x}^{\;i,j,\delta}) - \pi(\mathbf{x}^{\;j,i,\delta}) \qquad\forall i,j \in N,i<j,\tag{\ref{eq:fij}-a}
\end{equation*} 
and 
\begin{equation*}
  f_{0}(\mathbf{x}) := \sum_{k \in N}\; x_{k} - v(N).\tag{\ref{eq:fij}-b}
\end{equation*}
For any overdetermined system, an equivalent minimization problem is associated such that the set of global minima coincides with the solution set of the system. The solution set of such a minimization problem is the set of payoff vectors $\mathbf{x}$ which minimizes the function 
 \begin{equation}
    \label{eq:objfh}
  h(\mathbf{x}) := \sum_{\substack{i,j \in N\\ i < j}}\; (f_{ij}(\mathbf{x}))^2 + (f_{0}(\mathbf{x}))^2 \ge 0 \qquad\;\forall\,\mathbf{x} \in \mathbb{R}^{n}.
\end{equation}
For further details, see~\citet[Chap. 5 \& 6]{mei:13}. Then, one can establish the subsequent result:

\begin{corollary}[The Set of Global Minima Characterizes the Pre-Kernel (\citet{mei:13})]
 \label{cor:rep}
  For a TU game $\langle N,v \rangle$ with indirect function $\pi$, it holds that 
  \begin{equation}
    \label{eq:prkbyh}
    h(\mathbf{x}) = \sum_{\substack{i,j \in N \\ i < j}}\; (f_{ij}(\mathbf{x}))^2 + (f_{0}(\mathbf{x}))^2 = \min_{\mathbf{y} \in \mathcal{I}^{*}(N,v)}\; h(\mathbf{y}) = 0,
  \end{equation}
if, and only if, $\mathbf{x} \in \mathcal{PrK}(N,v)$.
\end{corollary}
Therefore, the pre-kernel set can be characterized as the solution set of function $h$. This result will be used to propose a comprehensible and more efficient routine to seek a pre-kernel element. Nevertheless, notice that the objective function $h$ of type~\eqref{eq:objfh} is not suitable to implement an efficient and simple method to single out a pre-kernel element. This is due to that the structural form of objective function $h$ is not well-behaved. The function is neither convex nor quadratic, nor it is differentiable everywhere. Apparently, it is possible to apply on function~\eqref{eq:objfh} a modified {\itshape Steepest Descent Method} to get a pre-kernel point. This was the proposed method by~\citet{mes:97}. In contrast, the proposed methods presented by~\citet[Algorithms 8.1.1-8.3.1]{mei:13} rely on the idea to identify smaller and more tractable optimization problems, from which a pre-kernel element can be determined. This preferred method is applicable by the observation that the objective function $h$ is composed of a finite family of quadratic and convex function (cf.~\citet[Proposition 6.2.2]{mei:13}). The proposed algorithms subsequently solving convex optimization problems which converge by a finite number of iteration steps to a pre-kernel point. We shall reexamine some crucial properties of the pre-kernel to identify a partition on the payoff space, that is, a partition of the domain of function $h$ into payoff equivalence classes, from which a unique convex and quadratic function can be obtained. As we understand these details, we discuss the method that is employed by the Algorithm~\ref{alg:01} to find an element from the pre-kernel set. 

To identify a partition of the domain of function $h$ into payoff equivalence classes, we first define the set of {\bfseries most effective} or {\bfseries significant coalitions} for each pair of players $i,j \in N, i \neq j$ at the payoff vector $\mathbf{x}$ by
\begin{equation}
  \label{eq:bsc_ij}
  \mathcal{C}_{ij}(\mathbf{x}):=\bigg\{S \in \mathcal{G}_{ij}\;\bigg\arrowvert\; s_{ij}(\mathbf{x},v) = e^{v}(S,\mathbf{x}) \bigg\}.
\end{equation}

When we gather for all pairs of player $i,j \in N, i \neq j$ all these coalitions that support the claim of a specific player over some other players, we have to consider the concept of the collection of most effective or significant coalitions w.r.t. $\mathbf{x}$, which we define as in~\citet[p. 315]{MPSh:79} by 
\begin{equation}
  \label{eq:bsc}
  \mathcal{C}(\mathbf{x}) := \bigcup_{\substack{i,j \in N \\ i \neq j} } \; \mathcal{C}_{ij}(\mathbf{x}).
\end{equation}

From the set of most effective or significant coalitions of a pair of players $i,j \in N, i \neq j$ at the payoff vector $\mathbf{x}$ the smallest cardinality over the set of most effective coalitions is defined as
\begin{equation}
  \label{eq:min_idx}
   \Phi_{ij}(\mathbf{x}) := \min \bigg\{\arrowvert S \arrowvert  \;\bigg\arrowvert\;  S \in \mathcal{C}_{ij}(\mathbf{x}) \bigg\}.
\end{equation}
Gathering all these sets having the smallest cardinality for all pairs of players $i,j \in N, i \neq j$, we end up with
\begin{equation}
  \label{eq:mi_coal}
  \Psi_{ij}(\mathbf{x}) := \bigg\{ S \in \mathcal{C}_{ij}(\mathbf{x}) \;\big\arrowvert \Phi_{ij}(\mathbf{x}) = \arrowvert S \arrowvert \bigg\}.
\end{equation}

With respect to an arbitrary payoff vector $\mathbf{x}$, the set of coalitions of smallest cardinality $\Psi_{ij}(\mathbf{x})$ which is minimized w.r.t. the lexicographical order $<_{L}$ is determined by
\begin{equation}
  \label{eq:sm_cidx}
  \mathcal{S}_{ij}(\mathbf{x}) := \bigg\{ S \in  \Psi_{ij}(\mathbf{x}) \;\bigg\arrowvert S <_{L} T\;\,\text{for all}\; S \neq T \in \Psi_{ij}(\mathbf{x}) \bigg\} \quad\forall i,j \in N, i \neq j.
\end{equation}
We denote this set as the {\bfseries lexicographically smallest most effective coalitions} w.r.t. $\mathbf{x}$.

After these preparative definitions, we are in a position to specify the set of lexicographically smallest most effective coalitions w.r.t. $\mathbf{x}$ through
\begin{equation}
  \label{eq:mec}
  \mathcal{S}(\mathbf{x}) := \bigcup_{\substack{i,j \in N \\ i \neq j} } \; \mathcal{S}_{ij}(\mathbf{x}). 
\end{equation}
This set will be denoted more concisely as the set of {\bfseries lexicographically smallest coalitions}. Given the correspondence $\mathcal{S}$ on the payoff space, we say that two payoff vectors $\mathbf{x}$ and $\mathbf{y}$ are equivalent w.r.t. the binary relation $\sim$ iff $\mathcal{S}(\mathbf{x}) = \mathcal{S}(\mathbf{y})$. This binary relation induces a partition on the payoff space. Having identified payoff equivalence classes, we can select an arbitrary payoff vector to get a unique quadratic and convex function. To see this, select payoff vector $\mathbf{x}$ from payoff equivalence class $[\vec{\gamma}]$, then we get the set $\mathcal{S}(\mathbf{x})$, from which a rectangular matrix $\mathbf{E}$ can be constructed through $\mathbf{E}_{ij}:= (\mathbf{1}_{S_{ji}} - \mathbf{1}_{S_{ij}}) \in \mathbb{R}^{n}, \;\forall i,j \in N, i < j$, and $\mathbf{E}_{0} := - \mathbf{1}_{N} \in \mathbb{R}^{n}$. Let $q =\binom{n}{2} + 1$ combining these $q$-column vectors, we can construct matrix $\mathbf{E}$ as a $(n \times q)$-matrix in $\mathbb{R}^{n \times q}$, which is given by
\begin{equation}
\label{eq:matE}
\mathbf{E} := [\mathbf{E}_{1,2}, \ldots,\mathbf{E}_{n-1,n},\mathbf{E}_{0}]  \in \mathbb{R}^{^{n \times q}}.
\end{equation}
A matrix $\mathbf{Q} \in \mathbb{R}^{n^2}$ can now be expressed as $\mathbf{Q} = 2 \cdot \mathbf{E} \; \mathbf{E}^{\top}$, a column vector $\mathbf{a}$ as $2 \cdot \mathbf{E} \; \vec{\alpha} \in \mathbb{R}^{n}$. Whereas, vector $\vec{\alpha} \in \mathbb{R}^{q}$ is specified through $\alpha_{ij} := (v(S_{ij}) - v(S_{ji})) \in \mathbb{R} \;\forall i,j \in N, i < j$ and $\alpha_{0} := v(N)$. Finally, the scalar $\alpha$ is given by $\Arrowvert \vec{\alpha} \Arrowvert^2$, whereas $\mathbf{E} \in \mathbb{R}^{n \times q}, \mathbf{E}^{\top} \in \mathbb{R}^{q \times n}$ and $\vec{\alpha} \in \mathbb{R}^q$.

From vector $\vec{\gamma}$, the set~\eqref{eq:mec} is constructed and then matrix $\mathbf{Q}$, column vector $\mathbf{a}$, and scalar $\alpha$ are induced, from which a quadratic and convex function can be specified through 
\begin{equation}
 \label{eq:objf2}
h_{\gamma}(\mathbf{x}) = (1/2) \cdot \langle\; \mathbf{x},\mathbf{Q} \,\mathbf{x} \;\rangle + \langle\; \mathbf{x}, \mathbf{a} \;\rangle + \mathbf{\alpha} \qquad \mathbf{x} \in \mathbb{R}^{n},
\end{equation}
In view of Proposition 6.2.2~\citet{mei:13}, function $h$, as defined by~\eqref{eq:objfh} is composed of a finite family of quadratic and convex functions of type~\eqref{eq:objf2}. For the details, we refer the interested reader to~\citet[Chap. 5 \& 6]{mei:13}. In accordance with Theorem 7.3.1 by~\citet[p. 137]{mei:13} a dual representation of the pre-kernel is obtained as a finite union of convex and restricted solution sets $M(h_{\gamma_{k}}, \overline{[\vec{\gamma}_{k}]})$ of a quadratic and convex function $h_{\gamma_{k}}$, that is, 

\begin{equation}
    \label{eq:nullsp}
\mathcal{PrK}(N,v) = \bigcup_{k \in \mathcal{J}^{\prime}}\;  M(h_{\gamma_{k}}, \overline{[\vec{\gamma}_{k}]}).
\end{equation}
where $\mathcal{J}^{\prime}$ is a finite index set such that $ \mathcal{J}^{\prime} : = \{k \in \mathcal{J}\, \arrowvert\, g(\vec{\gamma}_{k}) = 0\}$. In addition, $g(\vec{\gamma}_{k}) = 0$ is the minimum value of a minimization problem under constraints of function $h_{\gamma_{k}}$ over the closed payoff set $\overline{[\vec{\gamma}_{k}]}$. The solution sets $M(h_{\gamma_{k}}, \overline{[\vec{\gamma}_{k}]})$ are convex. Taking the finite union of convex sets gives us a non-convex set if $\arrowvert\,\mathcal{J}^{\prime}\,\arrowvert \ge 2$. Hence, the pre-kernel set is generically a non-convex set. For the class of convex games and three-person games we have $\arrowvert\,\mathcal{J}^{\prime}\,\arrowvert =1$, which implies that the per-kernel must be a singleton.

To this end, we consider a mapping that sends a point $\vec{\gamma}$ to a point $\vec{\gamma}_{\circ} \in M(h_{\vec{\gamma}})$ through 
\begin{equation}
 \label{eq:map_alg}
  \Gamma(\vec{\gamma}) := -\bigg(\mathbf{Q}^{\dagger}\; \mathbf{a} \bigg)(\vec{\gamma}) = -\bigg(\mathbf{Q}^{\dagger}_{\vec{\gamma}}\; \mathbf{a}_{\vec{\gamma}} \bigg)  = \vec{\gamma}_{\circ} \in M(h_{\vec{\gamma}}) \qquad\forall \vec{\gamma} \in \mathbb{R}^{n},
\end{equation}
where $\mathbf{Q}_{\vec{\gamma}}$ and $\mathbf{a}_{\vec{\gamma}}$ are the matrix and the column vector induced by vector $\vec{\gamma}$, respectively. Notice that matrix $\mathbf{Q}^{\dagger}_{\vec{\gamma}}$ is the pseudo-inverse of matrix $\mathbf{Q}_{\vec{\gamma}}$. In addition, the set $M(h_{\vec{\gamma}})$ is the solution set of function $h_{\vec{\gamma}}$. Under a regime of orthogonal projection this mapping induces a cycle-free method to evaluate a pre-kernel point for any class of TU games. We restate here Algorithm 8.1.1 of~\citet{mei:13} in a more succinctly written form through Method~\ref{alg:01}.

\footnotesize
\begin{algorithm}[H]
\DontPrintSemicolon
\SetAlgoLined
\KwData{Arbitrary TU Game $\langle\,N, v\,\rangle$, and a payoff vector $\vec{\gamma}_{0} \in \mathbb{R}^{n}$.}
\KwResult{A payoff vector s.t. $\vec{\gamma}_{k+1} \in \mathcal{PrK}(N,v)$.}
\Begin{
  \nlset{0} $k \longleftarrow 0, \quad\mathcal{S}(\vec{\gamma}_{-1}) \longleftarrow \emptyset$\;
  \nl Select an arbitrary starting point $\vec{\gamma}_{0}$\;
      \lIf{$\vec{\gamma}_{0} \notin \mathcal{PrK}(N,v)$}{Continue} \lElse{Stop}
  \nl Determine $\mathcal{S}(\vec{\gamma}_{0})$\;
      \lIf{$\mathcal{S}(\vec{\gamma}_{0}) \neq \mathcal{S}(\vec{\gamma}_{-1})$}{Continue} \lElse{Stop}
 \Repeat{$\mathcal{S}(\vec{\gamma}_{k+1}) = \mathcal{S}(\vec{\gamma}_{k})$}{
  \nl \lIf{$\mathcal{S}(\vec{\gamma}_{k}) \neq \emptyset$}{Continue} \lElse{Stop}
  \nl     Compute $\mathbf{E}_{k}$ and $\vec{\alpha}_{k}$ from $\mathcal{S}(\vec{\gamma}_{k})$ and $v$\;
  \nl  Determine $\mathbf{Q}_{k}$ and $\mathbf{a}_{k}$ from $\mathbf{E}_{k}$ and $\vec{\alpha}_{k}$\;
  \nl Calculate by Formula~\eqref{eq:map_alg} $\mathbf{x}$\;
  \nl $k \longleftarrow k+1$\;
  \nl $\vec{\gamma}_{k+1} \longleftarrow \mathbf{x}$\;
  \nl Determine $\mathcal{S}(\vec{\gamma}_{k+1})$\; 
        }
 }
 \caption{~Procedure to seek a Pre-Kernel Element}
 \label{alg:01}
 \end{algorithm}
\normalsize

\noindent \citet[Theorem 8.1.2]{mei:13} establishes that this iterative procedure converges toward a pre-kernel point. In view of~\citet[Theorem 9.1.2]{mei:13}, we even know that at most $\binom{n}{2}-1$-iteration steps are sufficient to successfully terminate the search process. However, we have some empirical evidence that generically a maximum of $n+1$-iteration steps are needed to determine an element from the pre-kernel set.

\section{Characterization of Efficiently Solvable Problems}
\label{sec:effsolvprob}
Resume that an algorithm is called to be of polynomial-time if its running time is upper bounded by a polynomial expression in the size of the input, as $O(n^{k})$, where $k$ is a constant. This means that the algorithm can be considered efficient or tractable. In contrast, an algorithm is said to be of exponential-time if the running time is upper bounded by $O(2^{poly(n)})$. For instance, if $poly(n)=n$, then we have $O(2^{n})$

To measure the input, at least the data of a cooperative game in the total number of players must be supplied to the model. Depending on the required precision and analysis of the model, more input data must be provided. For instance, arbitrary precision is required for a rational number model to encode the input as a finite string of bits. This implies that the number of elementary bit operations must be counted for a rational number model to grasp the underlying bit complexity. This is not needed for a real Random Access Machine (RAM) model, where it is assumed that the real number can be represented by infinite precision so that each operation on a real number takes constant time, regardless of precision.

For estimating the computational complexity of an optimization problem, we also have to care about strongly polynomial-time algorithms, for which it is stated that
 
\begin{enumerate}
\item it consists of the basic operations, viz., addition, subtraction, multiplication, division, as well as comparison; and
\item these basic operations are carried out on rationals/integers of size polynomially bounded in the dimension of the input; and 
\item the number of these operations is polynomially bounded in the dimension of the input.
\end{enumerate}

Roughly speaking, an algorithm is called strongly polynomial whenever its running time is a polynomial function of the number of $n$ the size of the input and independent of the magnitude of the input values.

Moreover, an algorithm is called {\bfseries NP}-hard if it refers to a problem that is at least as difficult to solve as the hardest problems in the complexity class {\bfseries NP} (class of nondeterministic polynomial-time problems). In particular, computing the (pre-)nucleolus or a (pre-)kernel point is {\bfseries NP}-hard in general. This is due to the fact that the best-known general algorithms for computing the (pre-)nucleolus involve, for instance, solving a sequence of linear programming problems, which have a running time that is exponential in the number of players $n$, since there are $2^{n}$ possible coalitions (exponential complexity).

\section{Ellipsoid Method of Computing the Nucleolus}
\label{sec:methellip}
\citet{faiglekernkuip:01} proposed a first polynomial-time algorithm for computing the nucleolus of a TU game while attacking an allocation in the intersection of the least-core with the pre-kernel. For TU games for which the pre-kernel contains exactly one core element, one can conclude that even the nucleolus was computed due to the following result: 

\begin{theorem}[Single Point Intersection of Core and Per-Kernel]
  \label{thm:intcoreprk}
  Let $\langle\, N, v\,\rangle$ be a TU game. If $\mathcal{C}(N,v) \cap \mathcal{PrK}(N,v) = \{\mathbf{x}^{*}\}$ is satisfied, i.e., the intersection of the core and pre-kernel just consists of a single point, then it holds $\nu(N,v) = \mathbf{x}^{*}$. Hence, the single intersection point is the nucleolus of the game $\langle\, N, v\,\rangle$.
\end{theorem}

Thus, the method of \citet{faiglekernkuip:01} for computing the nucleolus of a TU game can be applied to convex games, where the pre-kernel is a singleton, or minimum cost spanning tree games, where the core consists of a single point. But not to the class of permutation games, where the pre-kernel is inscribed in the least-core of the game (for more details, we refer the reader to~\citet{sol:14a,mei:14b}). 

In order to determine this exposed point in the intersection of the core with the pre-kernel of a TU game,~\citet{faiglekernkuip:01} imposed a combination of the ellipsoid method for solving sequences of linear programming problems with an approximation scheme for a (pre-)kernel point. This approximation scheme was first postulated by M. Maschler at the 1965 Jerusalem Conference on Game Theory, where the convergence was given as an open problem. In the context of the ellipsoid method, sufficiently many steps of Maschler's approximation scheme must be invoked to find suitable separating hyperplanes to apply the ellipsoid method.~\citet{stea:68} provided for Maschler's approximation scheme of the kernel a convergence proof using topological arguments while imposing the crucial condition that there exists an infinite number of transfers between a pair of players that are maximal. The required boundedness on the transfer to assure convergence was inherited by the compactness of the imputation set and the definition of the kernel. By~\citet[Theorem~4.1]{faiglekernkuip:01}, this result was extended to the pre-kernel without using any topological argument. Here, it was sufficient to establish that the imposed transfer scheme is bounded on the pre-imputation set, which is not a compact set. Nevertheless, their proof is logically flawed since the authors confound a proof by contradiction with a proof by contraposition (for further details, we refer to~\citet{mei:19}). Fortunately for them, their proof works out.

Notice that the main result of~\citeauthor{faiglekernkuip:01} in establishing a polynomial-time computation of the nucleolus relies on the crucial assumption that there exists a polynomial-time algorithm for computing all surpluses $s_{ij}(\mathbf{x},v)$ over all distinct pairs of players for any allocation $\mathbf{x} \in \mathbb{R}^{n}$. This is equivalent to imposing the assumption that the minimum excess $e^{v}_{\text{min}}(\mathbf{x})$ for any allocation $\mathbf{x} \in \mathbb{R}^{n}$ can also be efficiently computed by an algorithm, viz., in polynomial-time. This allows the authors to disregard the problematic issue of {\bfseries NP}-hardness.

Due to the known drawback of numerical instability in a real data model (unboundness of complexity; see below for more details) of the ellipsoid method,~\citeauthor{faiglekernkuip:01} studied the ellipsoid method in the field of rational numbers. This imposes several requirements on the input data that must be considered for a rational number model. Firstly, the imposed arbitrary precision of a rational number model requires knowledge of the encoding length of the largest number in the model that imposes an upper bound on the input size. In addition, it requires that the authors need to estimate the facet complexity of the polyhedron spanned by the system of linear inequalities of the linear programming problem that represents the least-core or core of the TU game. The facet complexity indicates the total number of bits required to encode the system of linear constraints (polyhedron) in the model. Moreover, the computational efficiency of the ellipsoid method stems from its ability to deal implicitly with an exponential number of constraints of a linear programming problem through a compact separation oracle. A separation oracle is a routine that, given a point, either confirms that the point is in the feasible set or provides a separating hyperplane. For real-valued linear programming problems, the separation oracle is just a matter of checking all the constraints, which is done efficiently. The complexity of the separation oracle needs to be accounted for in the overall runtime. In this respect, an oracle must query to compute and output the coalitional value as a rational number.

However, because of numerical instability issues of the ellipsoid method with floating-point arithmetic, the approach of~\citeauthor{faiglekernkuip:01} has practical limitations with floating-point arithmetic (for more details, see below).

\section{Sketch of the Ellipsoid Method}
\label{sec:skellip}

The algorithm for the ellipsoid method, pioneered by~\citet{khachiyan:79} for linear programming problems, works through iteratively creating smaller ellipsoids that contain a convex body like a polyhedron (feasible region). The process is started with an ellipsoid that contains the feasible region. In each subsequent step, the volume of the ellipsoid is reduced by a constant factor. This is done by calculating in a first step for each current ellipsoid its center of gravity. In the second step, all objective values that direct in the direction where all of them are becoming worse off are thrown away since the bad region cannot contain the optimum. By contrast, all of those objective values that are becoming better off are put in a new ellipsoid containing still the polyhedron, allowing the volume to be reduced by a constant factor.

In case the center of gravity of the current ellipsoid is feasible for the original linear programming problem, the bad values are thrown away as before. Otherwise, the center of gravity is not feasible; there is still a region with bad values. The bad region of the current ellipsoid is removed as before. A new ellipsoid is created that is smaller than the current ellipsoid. The number of iterations is governed by how many such reductions are needed to guarantee that the ellipsoid is either empty or contains a point in the feasible region (for more details, see~\citet[Chapter 3]{Groetlov:12}).

By the work of~\citet[Thm. 3.1]{Groetlov:81},~\citet[Chapter 6]{Groetlov:12}, it has been established that the problems of separation and optimization for any given polyhedron $\mathcal{P}$ are computationally equivalent since there is a polynomial-time algorithm to solve the separation problem of a polyhedron $\mathcal{P}$ if, and only if, there is a polynomial-time algorithm to solve the linear optimization problem for the polyhedron $\mathcal{P}$. In fact, they provide a sharpened separation result to even deal with non-full-dimensional polyhedrons. Thus, the problem of strictly separating any given polyhedron $\mathcal{P}$ by a point $\mathbf{x} \not\in \mathcal{P}$ s.t. $\mathbf{b}^{\top} \cdot \mathbf{x} < \mathbf{b}^{\top} \cdot \mathbf{y}$ for all $\mathbf{y} \in \mathcal{P}$ with $\mathbf{b} \in \mathbb{Q}^{n}$ holds true is computationally equivalent to solving the linear optimization problem related to polyhedron $\mathcal{P}$.

\section{Approximation Scheme of the Kernel}
\label{sec:approxkr}

The reader may confront our Approach~\ref{alg:01} based on the Fenchel-Moreau conjugation with the Standard Method~\ref{alg:06} from the literature for computing a kernel element while iteratively carrying out a sequence of side-payments between pairs of players to reduce the differences in the maximum surpluses between them. Thus, the algorithm is based on the idea to iteratively carrying out bilateral maximal transfers between pairs of players to approximate a kernel element, instead of a transfer scheme that simultaneously adjust the payoffs of all players involved in the game.

\tiny
\begin{algorithm}
\DontPrintSemicolon
\SetAlgoLined
\KwData{Arbitrary TU Game $\langle\,N, v\,\rangle$, and a payoff vector $\vec{\gamma}_{0} \in dom\,h$.}
\KwResult{A payoff vector s.t. $\vec{\gamma}_{k+1} \in \mathcal{K}(v)$.}
\Begin{
  \nlset{0} $k \longleftarrow 0$\;
  \nl Select an arbitrary starting point $\vec{\gamma}_{0}$.\;
  \nl Pre-specify a tolerance value $\epsilon$.\;
  \Repeat{$\frac{\delta^{*}}{v(N)} \le \epsilon$}{
      \lIf{$\vec{\gamma}_{k} \in dom\,h$}{Continue} \lElse{Stop}
  \nl    Compute the maximum surplus $s_{ij}(\vec{\gamma}_{k})$ for all $i,j \in N, i \neq j $ by~\eqref{eq:maxexc}.\;
  \nl    $\delta^{*} \longleftarrow \max_{ij \in N,i\neq j}\,\{ s_{ij}(\vec{\gamma}_{k}) - s_{ji}(\vec{\gamma}_{k}) \}$.\;
  \nl    Select a player pair $(i^{*},j^{*})$ s.t. $(i^{*},j^{*}) \in \{i,j \in N\,\arrowvert\; \delta^{*} = s_{ij}(\vec{\gamma}_{k}) - s_{ji}(\vec{\gamma}_{k})\}$.\;  

  \nl \lIf{$(\gamma^{j^{*}}_{k}-v(\{j^{*}\})) < \frac{\delta^{*}}{2}$ }{$\delta \longleftarrow (\gamma^{j^{*}}_{k}-v(\{j^{*}\})) $} \Else{$\delta \longleftarrow \frac{\delta^{*}}{2}$}
  \nl $\gamma^{i^{*}}_{k+1} \longleftarrow \gamma^{i^{*}}_{k} + \delta$.\;
  \nl $\gamma^{j^{*}}_{k+1} \longleftarrow \gamma^{j^{*}}_{k} - \delta$.\;
  \nl $k \longleftarrow k+1$.
   } 

}
 \caption{~\citeauthor{stea:68}'~Transfer Scheme to converge to a Kernel Element}   
\label{alg:06}
\end{algorithm}

\normalsize

We reproduce here~\citeauthor{stea:68}' Theorem in an alternative version, as given by~\citet{faiglekernkuip:01}.
\begin{theorem}[\citet{stea:68}]
  \label{thm:stea}
 Let $\{\vec{\gamma}_{0}, \vec{\gamma}_{1}, \vec{\gamma}_{2}, \ldots \}$ be a sequence of vectors such that $\vec{\gamma_{k}}$ arises from $\vec{\gamma}_{k^{\prime}-1}$ by a transfer~\eqref{eq:sidepyA}, and if an infinite number of these transfers are maximal, then $\vec{\gamma_{k}}$ converges to an element $\mathbf{x} \in \mathcal{K}(v)$.
\end{theorem}

The~\citeauthor{stea:68}' algorithm runs as follows: During each iteration step, the pair of players with the largest difference in the maximum surpluses is singled out to give them the opportunity to balance their claims. The player with the smallest maximum surplus makes a side-payment to the opponent who has a larger dissatisfaction measured in terms of the excess at the proposed sharing of the proceeds of mutual cooperation. We realize that at each step only a single pair of players has the opportunity to balance their claims with the implication that after each iteration step the unbalancedness in the bilateral claims will be reduced. In the limit, this approach converges to a kernel point. However, computing a kernel element requires an infinite number of maximal transfers as the side-payments carried out become arbitrarily small. These properties of his method have some negative consequences on the convergence speed for floating-point arithmetic.

\section{Ellipsoid Method: Runtime Complexity for Rational Data}
\label{sec:elliprat}
It is well-known that the runtime complexity for the ellipsoid method has a polynomial-time worst-case runtime complexity of $O(n^{6}\cdot L)$. For a linear programming problem in $n$ variables with rational coefficients, the runtime of the ellipsoid method is polynomial in the following factors: 

\begin{enumerate}
\item $n$: Representing the number of variables (dimension of the space).
\item $m$: Representing the number of constraints.
\item $L$: Representing the size of the input data, often measured as the number of bits required to encode the coefficients of the matrix $\mathbf{A}$ and vectors $\mathbf{b}$ and $\mathbf{c}$.  
\item Operations per iteration: Each iteration involves matrix and vector arithmetic. When the input is rational, these operations must be performed using multi-precision arithmetic, which is slower than fixed-precision arithmetic. The time for these operations depends polynomially on the bit size of the numbers involved.
\item Growth of bit length: The numbers representing the centers of gravity and matrices of the ellipsoids grow in size during the runtime of the algorithm. It can be established that the input bit length $L$ bounds this growth polynomially. 
\end{enumerate}

The combination of a polynomial number of iterations and polynomial-time operations per iteration (when accounting for multi-precision arithmetic) leads to an overall polynomial-time algorithm.

Note that the above analysis of algorithms dealing with rational numbers requires a computational model that can handle these factors with arbitrary precision. The standard for this is the Turing machine model, where the input is encoded as a finite string of bits. A rational number, $p/q$, is encoded by storing the binary representation of its integer numerator $p$ and denominator $q$. The size of a rational number is the total number of bits required for this representation. In terms of a rational number model, this means that every operation with a rational number depends on the precision and therefore cannot require constant time. 

For a feasibility problem where a convex set $K$ is either empty or contains a ball of radius $r$ and is contained in a ball of radius $R$, the number of iterations required is $O(n^{2}\log (R/r))$. In addition, for linear programming problems with integer coefficients, this translates to a number of iterations that is polynomial in $n$ and the size of the input $L$.

\section{Ellipsoid Method: Runtime Complexity for Real Data}
\label{sec:ellipreal}
When considering linear programming problems with real-valued data, the runtime complexity of the ellipsoid method is typically analyzed within a specific computational model, such as the real RAM (Random Access Machine) model. Unlike the Turing machine model, which deals with finite bit lengths, the real RAM model assumes that real numbers can be stored and manipulated with infinite precision in constant time. Though a real RAM model is Turing complete, which means that a real RAM model is in principle capable of performing all calculations that are feasible within a Turing machine model. In such a model, the runtime is not expressed in terms of the input's bit length $L$ but instead depends on the number of variables $n$, the number of constraints $m$, and the required precision, $\epsilon$. To summarize, the runtime of the ellipsoid method for the real RAM model is expressed in the following factors: 

\begin{enumerate}
\item $n$: Representing the number of variables.
\item $m$: Representing the number of constraints.  
\item $R$: Representing the radius of an initial ball known to contain the solution.
\item $r$: Representing the radius of a ball that is either contained in the feasible region or demonstrates that the feasible region is empty.
\item $\epsilon$: Representing the desired precision.  
\end{enumerate}

Then, the number of iterations required for the ellipsoid method to find an $\epsilon$-approximate solution for a convex optimization problem is proportional to $O(n^{2}\log (R/r\epsilon ))$. 

In each iteration, the algorithm performs a fixed number of operations, primarily matrix-vector and matrix-matrix multiplications, which can be done in polynomial-time in $n$. While the method's polynomial worst-case complexity holds in theory, this is generally not a good indicator of its practical performance. In practice, it is often outperformed by other methods, such as interior-point methods with a runtime complexity of $O(n^{3}\cdot L)$, which can be more robust to numerical issues. Even the simplex algorithm, though theoretically exponential, is regularly much faster in practice.

\section{Unbounded Complexity in the Real RAM Model}
\label{sec:unbcomp}
The theoretical analysis of the ellipsoid method in the real RAM model hides the practical challenges of using real-valued data with computers, which use finite-precision floating-point arithmetic. 

The ellipsoid method's runtime depends on the required precision, $\epsilon $. The number of iterations is proportional to $O(n^{2}\log (R/r\epsilon))$, where $R$ is the radius of the initial ellipsoid and $r$ is the radius of the largest ball contained in the feasible region. In the real RAM model, the concept of input size is based on the number of variables and constraints, but not on the precision of the numbers, because infinite precision is assumed, so that every operation with a real number takes constant time. As $\epsilon$ can be arbitrarily small, the logarithm term $\log (1/\epsilon)$ can be arbitrarily large, making the number of iterations and thus the runtime unbounded. For example, if one requires a solution with twice as many decimal places, the number of iterations increases. With real numbers, there is no inherent limit to how much precision might be required. This implies numerical instability issues with floating-point arithmetic.

When focusing on the theoretical analysis of the real RAM model, it assumes that real numbers can be represented and manipulated with perfect, infinite precision. As a consequence, the following problems arise when implementing this on a physical computer with finite-precision floating-point arithmetic:

\begin{enumerate}
\item Accumulation of error: It is causative that the ellipsoid method is sensitive to numerical precision issues. Its immanent iterative nature results in the fact that small errors from floating-point arithmetic can accumulate over many steps, potentially causing the algorithm to fail or produce an inaccurate result. This pathological behavior is attributable to the iterative process of repeatedly updating a matrix that describes the ellipsoid's shape and orientation. During each matrix updating, small rounding errors are introduced due to floating-point arithmetic. Over many iterations, these errors accumulate and can corrupt the matrix, causing the algorithm to fail or converge to an incorrect solution.
\item Loss of properties: This is causal because the core function of the algorithm relies on the matrix defining the ellipsoid being positive definite. Floating-point errors can cause this property to be lost during matrix updates, rendering the algorithm's mathematical foundations unsound during practical computations.
\end{enumerate}

It is essentially thanks to the research of~\citet{Traub:1982} that this connection was elucidated by demonstrating the theoretical unbounded complexity of the ellipsoid method in the real RAM model. They provided a $2\times 2$ linear programming example where the number of steps grows with the size of the input coefficients. In the integer model, the coefficients' bit length, $L$, is finite, so the number of steps is polynomial in $L$ and thus bounded. Moreover, in the real number model, the bit length of coefficients is not part of the input measure. As the number of bits required to represent the coefficients grows, the number of steps increases, demonstrating unbounded complexity. This counterexample proved that the ellipsoid algorithm is not polynomial in the real RAM model.

As a consequence, the ellipsoid method is in practice very slow and numerically unstable, especially for larger problems.

\section{Reduced Game and Submodular Function Minimization Method}
\label{sec:rgpsfm}
\citet{maggiorano:2025} conceived a different method for trying to compute the nucleolus for the class of convex games. However, their strategy to determine the nucleolus for this game class is, in the first stage, similar to the one invented by~\citet{faiglekernkuip:01}; they used the least-core as a catcher of the nucleolus. Yet, instead of using the ellipsoid method for computing the least-core value, they applied a submodular function minimization approach in testing the feasibility of the least-core constraints in polynomial-time. From the feasible region of the least-core, the search space is narrowed and bounded to approach by an alternative method the computation of the nucleolus for a convex game. This is referring to a simplified version of the Davis/Maschler reduced game that is induced by the convexity property of the underlying game. Hence, the authors decided against Maschler's approximation method for the (pre-)kernel, making a bet to disregard the computation of all surpluses $s_{ij}(\mathbf{x},v)$ over all distinct pairs of players for any allocation $\mathbf{x} \in \mathbb{R}^{n}$ at the cost of being forced to restrict the computation of the nucleolus on the domain of convex games and the risk of being forced to close a gap in their proposed method. By their suggestion for computing the nucleolus, the structural properties of convex games are exploited. This allows them to refer to the Davis/Maschler reduced game approach to single out the presumptive nucleolus of a convex game. In particular, the process iteratively solves a sequence of Davis/Maschler reduced games by leveraging submodular function minimization for solving the arising linear programming problems associated with the least-cores from the emerging Davis/Maschler reduced games.

So far so good, their intended strategy in computing the nucleolus of a convex game, if it would work; however, we shall observe in a while that this strategy fails. But first, let us examine some other interesting ingredients of their algorithm. 

In particular, although it is not necessary to query a separation oracle for checking whether a point is in the feasible set of the arising least-cores or to provide separating hyperplanes as with the ellipsoid method. The procedure of submodular function minimization requires, nevertheless, querying an oracle for computing and outputting coalitional values. On the positive side, the proposed submodular function minimization approach does not need any assumption made on the existence of any polynomial-time algorithm to assure a polynomial runtime in the nucleolus computation. This puts this model on a more sound and solid logical foundation in contrast to the method suggested by~\citet{faiglekernkuip:01}, thereby avoiding the {\bfseries Principle of Explosion}\footnote{This principle is also subsumed under ex falso sequitur quodlibet (e.f.s.q),~i.e., from falsehood, anything follows. If the antecedence (premise) is inconsistent, it cannot be a true statement. Then, drawing a conclusion is meaningless because anything is inferable. Hence, this says that an inconsistent model is useless.} if the assumption is void (for further details, we refer to~\citet{mei:19}). In addition, the model can be applied for infinite precision models, implying that neither knowledge of the encoding length of the largest number in the model nor knowledge about the total number of bits to encode the problem is necessary.

Next, we restate for the reader's convenience Algorithm 1 from the work of \citet{maggiorano:2025} that should, according to the claim of the authors, compute the nucleolus of a convex game. 

\tiny
\begin{algorithm}
\DontPrintSemicolon
\SetAlgoLined
\KwData{A convex TU Game $\langle\,N, v\,\rangle$, given through an evaluation oracle.}
\KwResult{Presumed nucleolus $\nu(N,v)$ of the game $\langle\,N, v\,\rangle$.}
\Begin{
  \nlset{0} Initialization: Set $k \longleftarrow 1$\;
  \nl Set $v^{k} \longleftarrow v$\;
  \nl Set $N^{k} \longleftarrow N$\;
  \nl Select $i \in N$\;
  \Repeat{i = n}{
      \While{$\arrowvert N^{k} \arrowvert > 1$}{
  \nl    Compute an essential coalition $S^{k}$ and the least-core value $\epsilon^{k}$ of game $\langle\,N^{k}, v^{k}\,\rangle$.\;
  \nl    Define $N^{k+1} \longleftarrow  S^{k}$ if $i \in  S^{k}$; otherwise, $N^{k+1} \longleftarrow N^{k} \backslash S^{k}$.\;
  \nl    Define $T^{k} := N^{k} \backslash N^{k+1}$, and $\vec{x}^{k} := \nu(N^{k},v^{k})$ (resume that $v^{k}(T^{k})$ is known).\;
  \nl    Define $v^{k+1} \longleftarrow v^{k}_{N^{k+1},\,\vec{x}^{k}}$ as $v^{k+1}_{N^{k+1},\,\vec{x}^{k}}(T) := \max\Big\{v^{k}(T),v^{k}(T \cup T^{k})-x^{k}(T^{k})\Big\}, \;\text{if}\; T \subset N^{k+1}$.\;  
  \nl    Set $k \longleftarrow k + 1$\;  
      }
  \nl $\nu_{i}(N,v) \longleftarrow v^{k}(N^{k})$
   } 
 }
 \caption{Presumed Algorithm for Calculating the Nucleolus based on RGP}
\label{alg:07}
\end{algorithm}

\normalsize

However, if we dive more deeply into the details of their purported method to compute the nucleolus and their associated proof, we have to recognize that their Method~\ref{alg:07} is not about the computation of the nucleolus for convex games. It is, in the best case only about its verification. In fact, it is a partial application of the Davis/Maschler reduced game property (RGP) to least-core elements of the convex game $\langle\,N, v\,\rangle$. This can be easily grasped by focusing our attention at level $k=1$. Then, the least-core value of the original convex game is specified in association with an arbitrary optimal allocation $\vec{x}$ that satisfies these constraints. Nevertheless, the authors claim that they have found the nucleolus of the game $\nu(N^{1},v^{1}) = \nu(N,v)$ in line 6 of Algorithm~\ref{alg:07}. Even by inspection of their purported proof, this claim is argumentatively fortified. Unfortunately, this argument used by the authors is obviously wrong. In this respect, we need to ask why one should proceed with their method to compute the nucleolus of a convex game when it was already found at the first stage, viz., $k=1$; this makes no sense. 

Firstly, to decipher the error of the procedure, observe that their proof to assess the correctness of Algorithm~\ref{alg:07} is based on an incorrect application of RGP, which does not lead to a correct computation of the nucleolus. For although the Davis/Maschler reduced game property (RGP) is attributed, for instance, to several solution concepts like the core and the pre-kernel (cf.~\citet{Pel:86b}), and by the work of~\citet{Sob:75} also to the pre-nucleolus, the authors solely attribute, through their Proposition 2.2, this property to the pre-nucleolus. 

In this context, it is well-known through~\citet{MPSh:72} that a Davis/Maschler reduced game $\langle\, S, v_{S,\,\vec{x}}\,\rangle$ is convex for a core allocation $\vec{x}$ whenever the default game $\langle\,N, v\,\rangle$ is convex. Thus, the convexity property of the Davis/Maschler reduced game is not only restricted to the nucleolus $\vec{x}:=\nu(N,v)$ of the original game $\langle\,N, v\,\rangle$ as the authors claim by their Proposition 2.3. It also holds for all reduced games that can be derived from a core allocation of a convex game, i.e., the statement holds for all $\vec{x} \in \mathcal{C}(N,v)$ -- including the nucleolus -- whenever $\langle\,N, v\,\rangle$ is convex. The same argument also applies to the formula applied by the authors in line 7 of Algorithm~\ref{alg:07}. Interestingly, although the author recognized these facts in the appendix of their work, they completely ignored the implications in their proofs.

Contrary to all facts from the literature, they apply, through Proposition 2.2, their own version of RGP. This states that the restriction of the computed nucleolus to the components of coalition $N^{k}$ is also the nucleolus of the reduced game $\langle\, N^{k}, v_{N^{k},\,\vec{x}}\,\rangle$. This is, of course, a true statement. However, this version does not contain the essentials for this game property, so the crucial implications of it are overlooked by~\citeauthor{maggiorano:2025} In particular, that the reduced game property is a consistency property (see Definition~\ref{def:hm_rgp0}).

Consistency is widely regarded as a desirable property of a solution concept in cooperative game theory, as it prevents any possible subsets of players from redistributing the payoffs of mutual cooperation while adhering to the underlying rules of distributive arbitration specified by the solution concept. Roughly speaking, a solution is considered consistent whenever it distributes the same payoff to the players of any appropriately defined reduced game as in the original game; that is to say, a solution is viewed as consistent under this specific rule. No subgroup of agents has an incentive to deviate from the original proposal and to play their own game in order to improve their situation. Consistency can be considered as the subgame perfection of a cooperative solution concept. In the sequel, we show that consistency cannot be used to deduce the original solution, at least not in the manner used by~\citeauthor{maggiorano:2025}, since they put aside that RGP holds even for all core allocations and not only for the nucleolus, culminating in the following statement by~\citeauthor{maggiorano:2025}:  

\begin{quote}
 Furthermore, by Proposition 2.2, the components of the nucleolus computed in the reduced games are equal to the corresponding components of the original nucleolus $\nu$, that is $\nu^{\,k}_{j} = \nu_{j}$ for all $j \in N^{\,k}$ (cited from~\citet[p.~7, first paragraph]{maggiorano:2025}).  
\end{quote}

Next, we will demonstrate that this view is incorrect by using the game example presented in the article by~\citet{maggiorano:2025}, which we will summarize in more detail in Example~\ref{exp:repnuc_cv}. In Example~\ref{exp:count_exp} below, however, we focus only on the essentials of the game example in order to refute the authors' view.

\begin{example}[Counter-Example to the above Statement]
  \label{exp:count_exp}
Be reminded that the game used in this article is convex with a least-core that is constituted by the line segment $$conv\{(3,3,2,2),(2,4,2,2)\},$$ and the nucleolus is given by $\nu(N,v) = (5/2,7/2,2,2)$, which is part of the least-core. 

Note that the algorithm proposed by~~\citeauthor{maggiorano:2025} does not exclude the selection of the vertex $(2,4,2,2)$, for instance, as it does not exclude any other point from the least-core. All allocations from the least-core are admissible selection points, which constitute the set of the nucleolus catcher.    

To refute the above statements, we shall use the second vertex of the least core of the original game to construct an associated reduced game for the coalition $S=\{2,3,4\}$; then we will compute the nucleolus of this reduced game and compare it with the restriction of the nucleolus from the original game to these components. Finally, we will show that the two nucleoli will not coincide, which is the expected result.

For this purpose, let us consider the three-person reduced game $v_{S,\mathbf{x}}$ with $S=\{2,3,4\}$ and $\mathbf{x} = (2,4,2,2)$ that is given by Table~\ref{tab:pkRGP1}:

\begin{center}
\begin{threeparttable}
{\footnotesize
\setlength{\tabcolsep}{.3cm}
\caption{Three-Person Reduced Game $v_{S,\mathbf{x}}$}
\begin{tabular}[c]{c c c c c c c c}
 \hline
Game & $\{2\}$ & $\{3\}$ & $\{4\}$ & $\{2,3\}$ & $\{2,4\}$ & $\{3,4\}$ &  $\{2,3,4\}$ \\[.1em] \hline
$v_{S,\mathbf{x}}$\tnote{a,b,c} &   $1$     &   $0$    &    $0$   &  $4$    &    $4$   &    $0$    &  $8$ \\[.1em] \hline\hline\
\end{tabular}
\vspace{-.5em}
\label{tab:pkRGP1}
\begin{tablenotes}
\item[a] Nucleolus: $(4,2,2)$
\item[b] Shapley Value: $(26,11,11)/6$
\item[c] $v_{S,\mathbf{x}}$ is convex.   
\end{tablenotes}
}
\end{threeparttable}
\end{center}

The components of the original nucleolus to $(2,3,4)$ are given by $(7/2,2,2)$, implying $ (4,2,2) \neq (7/2,2,2)$, which violates the quoted statement from the work of~\citet{maggiorano:2025}. Hence, this violates the reduced game property with respect to the nucleolus solution of the original game for $S=\{2,3,4\}$ at $\mathbf{x} = (2,4,2,2)$. This counter-example establishes the selection issue that arises in the algorithm proposed by these authors when making an inappropriate choice from the least-core. However, notice that the restriction $\mathbf{x}_{S} = (4,2,2)$ belongs to the core of the reduced game $v_{S,\mathbf{x}}$, and satisfies therefore RGP with respect to the core for $S=\{2,3,4\}$ at $\mathbf{x} = (2,4,2,2)$.\hfill$\Diamond$
\end{example}

This implies that the authors are faced with a serious selection issue, which would only be unproblematic when the least-core of a convex game would be a singleton, which cannot be assured in general, as we recognized by the above example (for more details, see Example~\ref{exp:repnuc_cv} in connection with Figure~\ref{fig:repnuc01}, for instance). That is to say, their argument that they have imposed on the nucleolus applies to all allocations within the core and, therefore, even to all allocations in the least-core. As we can see from Example~\ref{exp:count_exp}, the authors cannot guarantee that their algorithm correctly singles out the nucleolus of a convex game rather than a profane member of the core that is even an allocation of the least-core. At best, this can only be ensured while supplying the algorithm with the known nucleolus of the convex game or that the least-core of the game consists of a single point. By these arguments, we conclude that their method might be useful in the verification of a computed candidate of the nucleolus for a convex game, but not more. 

In this regard, the authors claimed to have provided a first strongly polynomial-time algorithm for computing the nucleolus of convex games. If anything, these authors can only be credited with having sketched an algorithm suitable for checking the nucleolus of a convex game in strongly polynomial time, but not for its computation. Since, according to the arguments above, this cannot be guaranteed with their flawed method. Even the suitability to verify the nucleolus needs further adjustment to the algorithm, as it requires knowledge about the maximum surpluses $s_{ij}(\mathbf{x},v)$ over all pairs of players. To simplify matters, we assume the suitability in the sequel to finalize our analysis.

Nevertheless, let us take a closer look at the runtime complexity of verifying the nucleolus of a convex game by their method. The estimated total runtime complexity of their approach is about $O(\text{Total\ Runtime})\\=O(n^{4}\cdot \text{SFM})$, where $\text{SFM}$ denotes a placeholder for a subroutine for submodular function minimization. That means that the overall algorithm for verifying the nucleolus for convex games makes $O(n^{4})$ calls to the SFM subroutine. One of the best-known subroutines for submodular function minimization at our time can be attributed to the work of~\citet{lee:2015}\footnote{Other polynomial-time methods from this line of research are attributable to~\citet{chakrabarty:2019,brand:2020}}. This algorithm of submodular function minimization has a runtime complexity of $O(n^{3} \cdot \log^{2}{n} \cdot \text{EO} + n^{4} \cdot \log^{O(1)}{n})$, where $\text{EO}$ denotes the time to evaluate the submodular function on any given subset. Implying that the overall runtime complexity of the method proposed by~\citet{maggiorano:2025} is given by $O(n^{7} \cdot \log^{2}{n} \cdot \text{EO} + n^{8} \cdot \log^{O(1)}{n})$, or as $\tilde{O}(n^{7}\cdot \text{EO}+n^{8})$ while ignoring logarithmic factors. Applying an even more favorable polynomial-time estimation of $\tilde{O}(n^{8})$ for their computational approach to test the nucleolus for convex games, this implies that an exponential-time algorithm with an upper bound given by $\tilde{O}(2^{n})$ is outperformed by their algorithm for convex games when $44$ players are involved. This is still too large a number for today's double-precision computers, so their approach cannot have any practical relevance. Even worse, when relying on the more unfavorable polynomial-time estimation of $O(n^{7} \cdot \log^{2}{n} \cdot \text{EO} + n^{8} \cdot \log^{O(1)}{n})$, let us expect that this turning point is achieved for their algorithm beyond the $51$-player threshold depending on a bad estimation of $\text{EO}$. This means that one can only expect that their algorithm outperforms an exponential-time algorithm with an upper bound of $O(2^{n})$ only for the case of quadruple-precision computers. Implying that their approach, though it has, of course, some theoretical flavor, has no practical value on today's double-precision computers, although they claim that their approach outperforms by a factor of $n^{3}$ the ellipsoid method proposed by~\citet{faiglekernkuip:01}. But this is also what one can expect from a simple test routine.

\section{Method by Indirect Function Representation}
\label{sec:mthdindfuc}

The Fenchel-Moreau-based approach of Algorithm~\ref{alg:01}, though originally conceived to compute a pre-kernel element, can, as a byproduct, also be applied to compute the pre-nucleolus for all game classes for which it is known that the pre-kernel is single-valued. That is to say, not only for convex games, but also for almost-convex games or veto-rich games, just to mention some important game classes with a single pre-kernel point. 

To delineate the essential features of Algorithm~\ref{alg:01}, we account for the fact that only $n(n-1)/2+1$ constraints are needed, instead of considering a linear problem with an exponentially long set of equations by the ellipsoid method. This drastically reduces the order of the constraint set from $O(2^{n})$ to $O(n^{2})$, so that the overdetermined system of equations can be reformulated into a continuous but non-differentiable objective function. Moreover, the domain of this objective function partitions into a finite set of payoff equivalence classes, i.e., convex regions. Within each of these equivalence classes, the complex, non-differentiable objective function that characterizes the pre-kernel simplifies into a regular quadratic and convex function, allowing us to iteratively solve a system of quadratic convex minimization problems associated with these classes, rather than a single, complex, non-convex problem. Thus, by means of payoff equivalence class, the original difficult problem becomes a finite sequence of well-behaved convex quadratic minimization problems. Establishing that the Fenchel-Moreau based approach provides a substantial simplification of the characterization of the pre-nucleolus for the class of games with a single-valued pre-kernel. At each iteration step, a solution vector of a quadratic optimization problem of type~\eqref{eq:objf2} is determined, allowing for the simultaneous diminishing of all differences in the excesses among the pair of players. Therefore, just a linear equation of the form $\mathbf{Q}\,\mathbf{x} = \mathbf{a}$ must be solved by infinite precision via the mapping of~\ref{eq:map_alg} with a polynomial-time worst-case runtime complexity of $O(n^{3})$.

Furthermore, under a regime of orthogonal projections, the search space is filtrated for updating the payoff vector such that convergence to a pre-kernel element occurs cycle-free. Hence, the application of orthogonal projections ensures that at each step, the updated payoff vector is moved toward a balanced state (where surpluses are equalized) within a subspace corresponding to the current equivalence class. The imposed filtration of the vector space of excess configurations leads to carrying out, instead of an infinite number of transfers, just a finite number of at most $(\binom{n}{2}-1)$ side-payments among the pair of players to successfully terminate the search process (cf.~\citet[Theorem 9.1.2]{mei:13}). Here, each iteration step on the payoff space is translated into the vector space of excess configuration while increasing the dimension of the search space at least by one. Typically, this process requires at most $(n+1)$ iterations for an $n$-player game, instead of the worst-case scenario of $\binom{n}{2}-1$ steps. Hence, a fast convergence is guaranteed, which is a significant improvement over iterative methods like that of Maschler's approximation method that can have unpredictable or slower convergence.  

By the Fenchel-Moreau-based approach to compute the pre-nucleolus of a game with a single-valued pre-kernel, one avoids relying on the least-core as a search space that contains the pre-nucleolus. Thereby, one can dispense with the ellipsoid method or submodular function minimization. Rather, the method focuses directly on the balancedness condition of the maximum surpluses condition, implying that no intermediate bound as a pre-nucleolus catcher is needed. It directly navigates toward a pre-kernel element without needing to first restrict the search to a smaller and bounded set. This implies that the approach is more efficient because it is a specialized, elegant, and highly optimized algorithm for the pre-kernel that can also be applied for the pre-nucleolus computation for the indicated game classes. It uses a dual characterization based on Fenchel-Moreau conjugation but applies it in a structured way that avoids the complexity and potential inefficiencies of more generic convex optimization methods of this type. Hence, this method is not a general-purpose sequential QP method. It is a highly tailored algorithm specifically for the pre-kernel computation. The power of this approach stems from the specific dual characterization, achieved through the indirect function as an increasing polyhedral convex function, which enables us to identify the structure of the pre-kernel.

Moreover, if it is similar to the work of~\citet{faiglekernkuip:01} assured that all maximum surpluses $s_{ij}(\mathbf{x},v)$ over all distinct pairs of players for any allocation $\mathbf{x} \in \mathbb{R}^{n}$ can be computed by a polynomial-time algorithm. Then, it is also possible to specify the set of lexicographically smallest most significant coalitions~\eqref{eq:sm_cidx} on the quotient space $dom\, h/\sim$ in polynomial-time. Since finding the smallest or largest item in an unsorted array is of linear time in the size of the array. Yet, the assumption made about the maximum surpluses conceals a crucial logical weakness, which will manifest itself as a logical deficiency once it proves to be invalid. Similar to the assumption introduced by~\citeauthor{faiglekernkuip:01}~to overcome the {\bfseries NP}-hardness issue of computing a pre-kernel point or the pre-nucleolus by their method. Nevertheless, we do not consider the assumption made on the maximum surpluses $s_{ij}(\mathbf{x},v)$ as critical. As the practical relevance of their computation was proven by our MATLAB software tool {\ttfamily MatTuGames} (cf.~\citet{mei:18}) for large numbers of $n$ up to $35$, exploiting successfully the physical limits of our computer resources with $3$ TB of main memory in a reasonable time.

Furthermore, no querying of an oracle in any form is necessary; this is already accomplished by the polynomial-time algorithm for determining all possible maximum surpluses $s_{ij}(\mathbf{x},v)$ at allocation $\mathbf{x}$. Therefore, there is absolutely no need to evaluate submodular functions or coalitional values by an oracle. In addition, the model is dealing with infinite precision instead of arbitrary precision. Therefore, neither any knowledge about the encoding length of the largest number in the model nor some knowledge about the total number of bits to encode the problem is required. To this end, what we consider a major advantage over the discussed algorithms from above is that no restriction to narrow and bound the search space by a pre-nucleolus catcher is needed. Thus, no additional structure of the problem needs to be evaluated or a separation oracle needs to be queried. It should be evident that exploiting too much structure of a problem must have negative side effects on the performance of an algorithm. In this context, there is also no need to exploit any form of combinatorial property like that of extended polymatroids, or a structural property of a game class like convexity. This is owed to the fact that the regime of orthogonal projection provides for the filtration of the search space naturally and bounds the number of iteration steps to converge to the pre-nucleolus. Then, Algorithm~\ref{alg:01} is reduced to an algorithm that sends, via the mapping~\eqref{eq:map_alg}, a point $\vec{\gamma}$ to a point $\vec{\gamma}_{\circ} \in M(h_{\vec{\gamma}})$, i.e., into the solution set of a quadratic optimization problem. This can be accomplished by solving a linear equation with a runtime complexity of $O(n^{3})$, which indicates a polynomial-time algorithm. 

To compare the algorithms, we have to recognize that the proposed method by~\citet{maggiorano:2025} might only be applicable for verifying the pre-nucleolus of convex games after implementing some decisive modifications. Obviously, their flawed approach does not even apply to TU games, which are not convex. But when focusing on convex games, we know that the pre-kernel coincides with the pre-nucleolus in this game class. Hence, computing the sole pre-kernel point, one has found the pre-nucleolus, whichever method one has used. Therefore, whenever one has found the pre-nucleolus for convex games by Algorithm~\ref{alg:01} solving a sequence of linear equations of type~\eqref{eq:map_alg} with a runtime complexity of $O(n^3)$ and at most $n+1$ iteration steps (cf.~Algorithm 8.1.1 of~\citet[Chapter 8]{mei:13}), one has scaled much better than solving an iterative process of linear programming problems for nothing having runtime complexity of $\tilde{O}(n^{7}\cdot \text{EO}+n^{8})$ with $2 \cdot n - 2$ iteration steps. In contrast to the above, this algorithm outperforms an exponential-time algorithm with an upper bound of $O(2^{n})$ when $10$ players are involved, which demonstrates even its practical relevance on modern computer architecture.

\section{Replication of the Pre-Nucleolus}
\label{sec:repprk}

As we mentioned above, the Fenchel-Moreau-based algorithm is even useful for the computation of the pre-nucleolus whenever it is known that the pre-kernel is a single point, i.e., it coincides with the pre-nucleolus. Even though it was conceived as a highly tailored algorithm specifically for the pre-kernel computation. At first glance, this seems to significantly limit the application of this method in determining the pre-nucleolus if we do not possess this kind of knowledge. However, we shall demonstrate that this method has a much broader flexibility than one might expect, as we resume with the replication results of~\citet{mei:23}. In this context, replication means showing that a specific solution point from one game can also be found as a solution in a different but related game, even when the parameters of the original game are modified.  

Using the above structure induced by the Fenchel-Moreau-based characterization of the pre-kernel, \citet{mei:23} could establish a major replication result for a single-valued pre-kernel. This is to say, for the game classes where the pre-kernel coincides with the pre-nucleolus of the game. Roughly speaking, this study explores how a specific point within the pre-imputation set, i.e., the pre-nucleolus of the game, can be replicated whenever the pre-kernel consists of a single point. Turned differently, it can be shown that this point is even the sole pre-kernel point and, therefore, the pre-nucleolus, in a different but related game under certain conditions. Note that such a related game need not have the same game properties as the original game, though it is established that it possesses the same single-valued pre-kernel. 

Apart from the single-valuedness of the pre-kernel, one of these specific conditions leverages that the payoff equivalence class that contains the single pre-kernel point must be full-dimensional, meaning that the payoff equivalence class has a non-empty interior. This non-empty interior condition of the unique pre-kernel point allows inscribing a full-dimensional ellipsoid within the payoff equivalence class from which a null space in the game space can be identified. Through these variation bounds, the range within the game space is specified that leaves the pre-kernel element invariant against the group action on the set of equivalent ordered bases. This means within this range, the game parameter can be varied without affecting the pre-kernel properties of this payoff vector. To put it differently, a variation of the game parameter within these bounds does not destroy the pre-kernel properties of the solution from the default game. Implying that these bounds specify a redistribution of the bargaining power among coalitions while supporting the selected pre-imputation still as a pre-kernel point. Even though the values of the maximum surpluses have been varied, the set of most effective coalitions remains unaltered by the parameter change. This indicates that a bargaining outcome related to this specific pre-kernel point remains stable against a variation in the game parameter space, and obstruction is held to account. Hence, a set of related games can be determined that are linearly independent and possess the selected pre-kernel element of the default game as well as a pre-kernel point. And a fortiori, as its pre-nucleolus.

By~\citet{mei:23}, these results have been generalized. There it was established that even on the convex hull comprising the default and related games in the game space, the pre-kernel must be a single point and is identical with the element specified by the default game. Furthermore, the pre-kernel correspondence restricted on this convex subset in the game space must be single-valued and consequently continuous. Hence, the nucleolus is an invariant (stable) solution on a subdomain of the game space following the pre-kernel fairness standards (axioms), and compliance is reality and obstruction held to account.

To deepen the understanding of compliance on non-binding agreements, we refer the reader to~\citet{mei:17f}, where, from a Cournot situation, four cooperative game models are discussed, each of which represents different aspiration levels of partners involved in a negotiation process of splitting the monopoly proceeds. The bargaining difficulties are demonstrated that might arise when agents are not acting self-constraint, and what consequences this imposes on the stability of a fair agreement. 

\begin{example}
  \label{exp:repnuc_cv}
To present for the reader a summary of the ease and the elegance of a Fenchel-Moreau-based approach for computing the pre-nucleolus of a game with a single-valued pre-kernel, we shall consolidate this knowledge in a manually derived step-by-step procedure for a game example that we have borrowed from the literature. Besides the pre-nucleolus computation, we also demonstrate the flexibility of this approach by discussing the associated replication results. By doing so, we consider a four-person convex game, which we have picked up from the work of~\citet{maggiorano:2025}. This game is specified by
\begin{equation*}
  v(S):=
  \begin{cases}
   3  & \text{if}\; S = \{1,2\},\{2,3\},\{2,3,4\}, \\  
   6  & \text{if}\; S = \{1,2,3\},\{1,2,4\};  \\
   10 & \text{if}\; S = N, \\
   0  & \text{otherwise},\\
  \end{cases}
\end{equation*}

\noindent for all $S \subseteq N$ with $N=\{1,2,3,4\}$. To visualize some results, in the Figure~\ref{fig:repnuc01}, we have plotted the core of game $v$ as a yellow polytope with its $11$ vertices in connection with the imputation set (skeleton-like triangle). Secondly, the least-core as a line segment, which is given by the convex hull $conv\,\{\{3,3,2,2\},\{2,4,2,2\}\}$, is included in the graphic. Thirdly, the Shapley value (blue enlarged dot), the modiclus (green enlarged dot), and the pre-nucleolus (red enlarged dot), which is identical to the nucleolus, are drawn. The least-core has an $\epsilon$-value of $-2$, which is indicated in the headline of the figure. Furthermore, all mentioned point solutions are core members. The nucleolus is in the center of the least-core, i.e., the line segment spanned by the extreme points of the least-core.\footnote{The figure has been generated with our Mathematica Package TuGames implemented within \citet{mei:10a}.}. 

\begin{figure}[ht!]
\centering
\ifpdf
    \includegraphics[height=12cm, width=12cm]{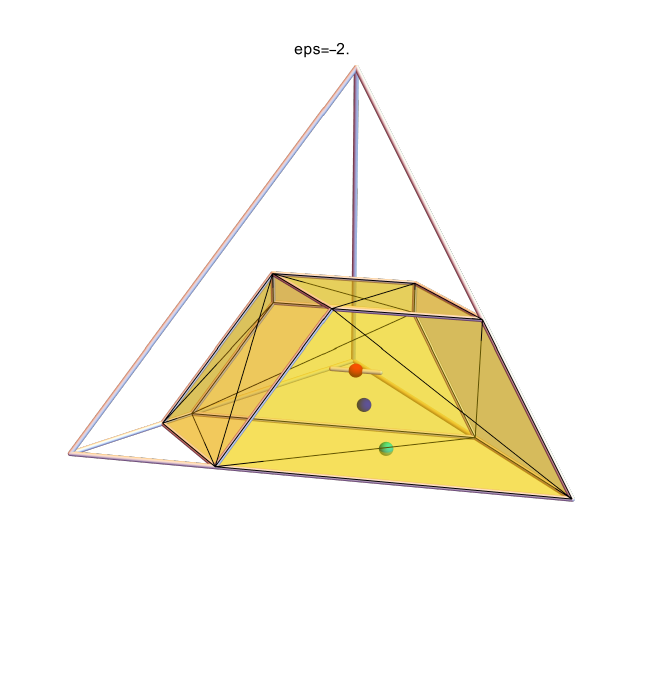} 
\else
    \includegraphics[height=9cm, width=12cm]{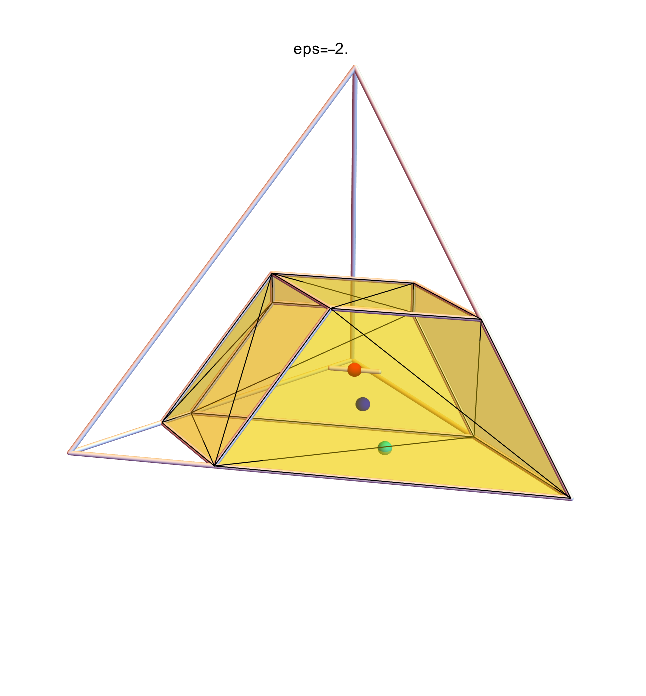} 
\fi
    \caption{The Imputation Set, Core, Least-Core, Nucleolus, Shapley value and Modiclus of Game $v$} 
    \label{fig:repnuc01}
\end{figure} 

\pagebreak

{\scriptsize
\begin{center}
\begin{ThreePartTable}
  \begin{TableNotes}
   \item[a] \label{tn:a} Pre-Kernel and Pre-Nucleolus: $(5/2,7/2,2,2)$
   \item[b] \label{tn:b} Modiclus: $\varsigma^{*}(N,v) = (5/2,5,5/2,0)$
   \item[c] \label{tn:c} Shapley Value: $\phi(N,v) = (11/4,17/4,7/4,5/4)$
   \item[d] \label{tn:d} CV: Convex Game.
   \item[e] \label{tn:e} SCV: Semi Convex Game.
   \item[f] \label{tn:f} CR: Game with non-empty Core.     
   \item[g] Note: Computation performed with MatTuGames (cf.~\citet{mei:18}).
  \end{TableNotes}
\begin{longtable}[c]{*{3}{cccccccccc}}
\caption[List of Games\tnote{f}~~which possess the same unique Pre-Kernel\tnote{a}~~as~$v$]{List of Games\tnote{f}~with the same single-valued Pre-Kernel (Nucleolus)\tnote{a}~~as~$v$} \label{grid_mlmmh} \\[.3em]

\hline \multicolumn{1}{c}{} &\multicolumn{1}{c}{} & \multicolumn{1}{c}{} & \multicolumn{1}{c}{} & \multicolumn{1}{c}{$\mu=-1.5$} & \multicolumn{1}{c}{} & \multicolumn{1}{c}{} & \multicolumn{1}{c}{} & \multicolumn{1}{c}{} \\ 
\multicolumn{1}{c}{Game} &\multicolumn{1}{c}{$\{1\}$} & \multicolumn{1}{c}{$\{2\}$} & \multicolumn{1}{c}{$\{1,2\}$} & \multicolumn{1}{c}{$\{3\}$} & \multicolumn{1}{c}{$\{1,3\}$} & \multicolumn{1}{c}{$\{2,3\}$} & \multicolumn{1}{c}{$\{1,2,3\}$} & \multicolumn{1}{c}{$\{4\}$} \\[.3em] \hline\hline 
\endfirsthead

\multicolumn{10}{c}%
{{\bfseries \tablename\ \thetable{} -- continued from previous page}} \\
\hline \multicolumn{1}{c}{} &
\multicolumn{1}{c}{} & 
\multicolumn{1}{c}{} & 
\multicolumn{1}{c}{} & 
\multicolumn{1}{c}{$\mu=-1.5$} & 
\multicolumn{1}{c}{} & 
\multicolumn{1}{c}{} & 
\multicolumn{1}{c}{} & 
\multicolumn{1}{c}{} \\
\multicolumn{1}{c}{Game} &
\multicolumn{1}{c}{$\{1,4\}$} &
\multicolumn{1}{c}{$\{2,4\}$} &
\multicolumn{1}{c}{$\{1,2,4\}$} &
\multicolumn{1}{c}{$\{3,4\}$} &
\multicolumn{1}{c}{$\{1,3,4\}$} &
\multicolumn{1}{c}{$\{2,3,4\}$} &
\multicolumn{1}{c}{$N$} &
\multicolumn{1}{c}{~CV~\tnote{d}} &
\multicolumn{1}{c}{~SCV~\tnote{e}} &
\multicolumn{1}{c}{~CR~\tnote{f}} & \\[.3em] \hline\hline 
\endhead

\hline \multicolumn{9}{r}{{Continued on next page}} \\ \hline\hline
\endfoot
\insertTableNotes\\
\hline \hline
\endlastfoot
$v$    &   0       &    0      &   3       &      0     &    0      &     3      &     6     &       0  \\
$v_{1}$ &  -16/27   &  -39/46   &  272/107   &   1/45    &  -38/83   &   65/27    &   271/45  &   1/45 \\
$v_{2}$ &  -21/79   &   15/29   &   125/31   &   3/47    &  -28/95   &   216/79   &   285/47  &   3/47 \\
$v_{3}$ &  -21/52   &   7/52    &  381/143   &  -23/32   &  -50/31   &   135/52   &   169/32  &   -23/32 \\
$v_{4}$ &  -13/43   &  14/139   &  169/64    &   -9/52   &  15/74    &   116/43   &   303/52  &   -9/52   \\
$v_{5}$ &  -13/43   &  14/139   &  169/64    &   -9/52   &  15/74    &   116/43   &   303/52  &   -9/52   \\
$v_{6}$ &  -13/43   &  14/139   &  169/64    &   -9/52   &  15/74    &   116/43   &   303/52  &   -9/52  \\
$v_{7}$ &  -6/61    &   2/61    &  273/97    &    4/23   &   1/48    &   177/61   &   142/23  &    4/23  \\
$v_{8}$ &  -6/61    &   2/61    &  273/97    &    4/23   &   1/48    &   177/61   &   142/23  &    4/23  \\
$v_{9}$ &  -6/61    &   2/61    &  273/97    &    4/23   &   1/48    &   177/61   &   142/23  &    4/23  \\
$v_{10}$ & -6/61    &   2/61    &  273/97    &    4/23   &   1/48    &   177/61   &   142/23  &    4/23   \\
\pagebreak
$v$     &    0      &    0       &    6     &    0       &      0     &     3     &  10  & Y & Y  & Y \\
$v_{1}$  &  -33/79   &   -35/52   &  271/45  &  -3/37     &   -47/88   &  317/136  &  10  & N & Y  & Y\\    
$v_{2}$  &  -83/152  &    9/38    &  285/47  &  23/65     &   -7/57    &  93/32    &  10  & N & Y  & Y \\    
$v_{3}$  &  -96/95   &  -17/36    &  169/32  &  -79/68    &   -237/142 &  94/37    &  10  & N & Y  & Y \\    
$v_{4}$  &  -61/45   &  29/53     &  303/52  &  -7/20     &    -6/7    &  91/29    &  10  & N & Y  & Y \\    
$v_{5}$  &   13/90   &  -20/21    &  303/52  &  -7/20     &    9/14    &  95/58    &  10  & N & Y  & Y \\    
$v_{6}$  &   13/90   &  29/53     &  303/52  &  -7/20     &    9/14    &  91/29    &  10  & N & Y  & Y \\    
$v_{7}$  &   4/29    &  25/93     &  142/23  &  -81/89    &    -19/25  &  123/58   &  10  & N & Y  & Y \\    
$v_{8}$  &   4/29    &  25/93     &  142/23  &  23/39     &    -19/25  &  105/292  &  10  & N & Y  & Y \\    
$v_{9}$  &   4/29    &  25/93     &  142/23  &  23/39     &    37/50   &  123/58   &  10  & N & Y  & Y \\    
$v_{10}$ &   4/29    &  25/93     &  142/23  &  23/39     &    37/50   &  105/29   &  10  & N & Y  & Y \\ \hline\hline
\label{tab:rpl_nuccvtug}
\end{longtable}
\end{ThreePartTable}
\end{center}
}

Further inspection reveals that this game has a veto-player, who is player two. A property that is much easier to test than convexity, especially for large values of $n$. Hence, apart from convexity, the game is even a veto-rich game, which gives an additional cause that the pre-kernel must coalesce with the pre-nucleolus (cf.~\citet{arinfelt:97}), whereas its outcome is given by the point $\nu(v) = \mathcal{PK}(v)= (5/2,7/2,2,2)$. Obviously, the set of lexicographically smallest most significant coalitions $\mathcal{S}(\nu(v))=\{\{1\},\{3\},\{4\},\{2,3\},\{1,2,3\},\{1,2,4\}\}$ is balanced with weights given by $\{1,1,2,1,1,1\}/3$. From this set, a boundary vector $\vec{b}=(5/2,2,2,11/2,8,8)$ is obtained by $\nu(v)(S)$ for $S \in \mathcal{S}(\nu(v))$. Define matrix $\mathbf{A}$ by $[\mathbf{1}_{S}]_{S \in \mathcal{S}(\nu(v))}$, then the solution of the system $\mathbf{A}\,\mathbf{x} = \vec{b}$ reproduces the pre-nucleolus. Moreover, this imputation is even an interior point of the payoff equivalence class composed by this collection. To see this, select an $\epsilon > 0$ sufficiently small and let $\mathbf{z} \in \mathbb{R}^{4}$ s.t. $z(N) = 0$; then one can establish that $\mathcal{S}(\nu(v) + \epsilon\,\mathbf{z}) = \mathcal{S}(\nu(v))$ holds. Thus, the non-empty interior condition is valid. Due to~\citet[Theorem 4.1]{mei:23}, a redistribution of the bargaining power among coalitions can be obtained while supporting the imputation $(5/2,7/2,2,2)$ still as a pre-kernel element for a set of related games.  In order to get a null space $\mathcal{N}_{\boldsymbol{\EuScript{W}}}$, we set the parameter $\mu$ to $-1.5$. In this case, the rank of matrix $\boldsymbol{\EuScript{W}}$ must be equal to $4$, and we could derive $10$ from at most $11$-linear independent games, which replicate the point $(5/2,7/2,2,2)$ as a pre-kernel element. Moreover,~\citet[Theorem 4.4]{mei:23} even states that this point is also the sole pre-kernel point; hence, the pre-kernel coincides with the pre-nucleolus for these games.

From Table~\ref{tab:rpl_nuccvtug}, we realize that the pre-nucleolus is stable against an induced variation in the game parameter. More precisely, the pre-nucleolus is invariant against a parameter variation within the specific interval of $\mu \in (-1.5, 1.4)$. To get an idea about the underlying bargain stability, we just mention here that the convex hull of the collection of games $\{v, \ldots, v_{10}\}$ form a subset in the game space -- the game setting has not changed, and with it not the figure of argumentation during the bargaining process; only the bargaining power of the coalitions has been changed without affecting the outcome of the game. Within this convex set, each TU game has exactly the above element as its sole pre-kernel point. This implies that we have identified a stable bargaining scenario where a settlement of an agreement is not problematic while referring to the principles of distributive justice (axioms) related to the pre-kernel. The agreement based on these principles remains stable after having varied the parameter set within the specified range. Therefore, the selected agreement point cannot be obstructed by the principles that describe the pre-kernel within this subset of the game space. Though the figure of argumentation has not referred to the norms of distributive arbitration based on the nucleolus, the nucleolus as a bargaining outcome is stable due to its coincidence with the pre-kernel, and the incentive to cooperate remains valid.

To illustrate how Algorithm~\ref{alg:01} works to compute the pre-nucleolus for games with a single-valued pre-kernel, let us focus on the pre-selected efficient payoff vector $\mathbf{y}_{0}=(10,0,0,0)^{\top}$ to see how we can apply this method for our specific example. By doing so, we look at the maximum surpluses for all pairs of players. For any pair of players $i,j \in N, i\neq j$, the maximum surplus of player $i$ over player $j$ with respect to any pre-imputation $\mathbf{x}$ is given by the maximum excess at $\mathbf{x}$ over the set of coalitions containing player $i$ but not player $j$. Thus,

\begin{equation*}  
  s_{ij}(\mathbf{x},v):= \max_{S \in \mathcal{G}_{ij}} e^{v}(S,\mathbf{x}) \qquad\text{where}\;  \mathcal{G}_{ij}:= \{S \;\arrowvert\; i \in S\; \text{and}\; j \notin S \}. 
\end{equation*} 

The expression $s_{ij}(\mathbf{x},v)$ describes the maximum amount at the pre-imputation $\mathbf{x}$ that player $i$ can gain without the cooperation of player $j$. 

From the excess vector $exc(\mathbf{y}_{0})$, we now get the subsequent set of lexicographically smallest coalitions for each pair of players

\begin{equation*}
\mathcal{S}(\mathbf{y}_{0}) = \{\{1\}, \{1,2,4\}, \{1, 2, 3\}, \{2\}, \{2,3\}, \{2,3\}, \{2,3\},\{2,3\}, \{3\}, \{2,3,4\},\{4\},\{4\}\},
\end{equation*}
whereas the order of the pairs of players in $\mathcal{S}(\mathbf{y}_{0})$ is given by $$\{(1,2),(1,3),(1,4),(2,3),(2,4),(3,4),(2,1),(3,1),(3,2),(4,1),(4,2),(4,3)\}.$$ 

To see how the selection process for the collection $\mathcal{S}(\mathbf{y}_{0})$ works, let us exemplarily consider the pair of players $(1,2)$. Focusing on this pair, we single out these coalitions that support the claim of player $1$ without counting on the cooperation of player $2$. These are the coalitions $\{\{1\},\{1,3\}, \{1,4\},\{1,3,4\}\}$ all having an excess of $-10$. Thus, all coalitions have maximum surplus, and the collection has a cardinality larger than $1$. If the collection has not a sole element, we have to apply two steps to find the lexicographically smallest coalition for the pair $(1,2)$. For the preceding case, both steps coincide, however. In order to proceed, we determine all the coalitions that have the smallest cardinality, which is coalition $\{1\}$. Since this set is unique, we have also determined the coalition that has a lexicographically minimum value in step one. Otherwise, at step two, we have to single out the lexicographically smallest coalition from an ambiguous collection. To observe how we have to proceed if this set is not unique, let us hypothetically assume for this purpose that the excess of coalition $\{1\}$ is smaller than $-10$. This implies that coalition $\{1\}$ is deleted from the list since it does not have maximum surplus. Hence, the collection has been modified to $\{\{1,3\},\{1,4\},\{1,3,4\}\}$. All of them have maximum surpluses, the smallest cardinality of a coalition from this collection is now $2$, and we single out the coalitions $\{\{1,3\},\{1,4\}\}$ and take finally the lexicographical minimum, which is $\{\{1,3\}\}$. 

For the reverse pair $(2,1)$, we find out that coalition $\{2,3\}$ supports best the claim of player $2$ without taking into account the cooperation of player $1$. Proceeding the same way for the remaining pairs, then we derive a matrix $\mathbf{E}$ by $\mathbf{E}_{ij}=\mathbf{1}_{S_{ji}} - \mathbf{1}_{S_{ij}}$ for each $i,j \in N, i < j$, and $\mathbf{E}_{0}=\mathbf{1}_{N}$. Notice that $\mathbf{1}_{S}:N \mapsto \{0,1\}$ is the characteristic vector given by $\mathbf{1}_{S}(k):=1$ if $k \in S$; otherwise, $\mathbf{1}_{S}(k):=0$. Then, matrix $\mathbf{E}$ is defined by 

\begin{equation*}
\mathbf{E} := [\mathbf{E}_{1,2}, \ldots,\mathbf{E}_{3,4},\mathbf{E}_{0}]  \in \mathbb{R}^{^{4 \times 7}}.
\end{equation*}

We realize that vector $\mathbf{E}_{1,2}$ is given by $(0,1,1,0)^{\top}-(1,0,0,0)^{\top}=(-1,1,1,0)^{\top}$ and $\mathbf{E}_{0}=(1,1,1,1)^{\top}$. A similar procedure applies to the remaining pairs of players; matrix $\mathbf{E}$ is quantified by
\begin{equation*}
\mathbf{E}=
\begin{bmatrix}[r]
  -1 & -1 & -1 &  0 &  0 &  0 &  1 \\
   1 &  0 &  0 & -1 & -1 & -1 &  1 \\
   1 &  1 &  0 &  1 & -1 & -1 &  1 \\
   0 & -1 &  1 &  0 &  1 &  1 &  1 
\end{bmatrix}.
\end{equation*}

A column vector $\mathbf{a}$ can be obtained by $\mathbf{E} \; \vec{\alpha} \in \mathbb{R}^{n}$, whereas the vector $\vec{\alpha}$ is given by $$\alpha_{ij} := (v(S_{ji}) - v(S_{ij})) \in \mathbb{R},$$ for all $i,j \in N, i < j $, and $\alpha_{0} := v(N)$. Therefore, vector $\vec{\alpha}$ is given by $$(3,-3,-3,0,-3,-3,10)^{\top}.$$.

From this matrix, we construct matrix $\mathbf{Q}$ by $\mathbf{E} \; \mathbf{E}^{\top}$. Inserting its numbers, matrix $\mathbf{Q}$ is specified by 

\begin{equation*}
\mathbf{Q}=
\begin{bmatrix}[r]
   4 &  0 & -1 &  1 \\
   0 &  5 &  3 & -1 \\
  -1 &  3 &  6 & -2 \\
   1 & -1 & -2 &  5  
\end{bmatrix}.
\end{equation*}
The column vector $\mathbf{a}$ is given by $(13,19,16,4)^{\top}$.

Solving this system of linear equations $\mathbf{Q}\;\mathbf{x} - \mathbf{a} = \mathbf{0}$, or alternatively $\mathbf{E}^{\top}\;\mathbf{x} - \vec{\alpha} = \mathbf{0}$, we get as a solution $\mathbf{y}_{1}=(128/37,98/37,91/37,60/37)^{\top}$, which is not a pre-imputation. Since the corresponding excess vector is given through
\begin{equation*}
  \begin{split}
   exc(\mathbf{y}_{1}) & =(0,-128/37,-98/37,-115/37,-91/37, -219/37, -78/37,-95/37,\\ & -60/37,-188/37,-158/37, -64/37,-151/37,-279/37,-138/37,-7/37). 
  \end{split}
\end{equation*}

We observe that the excess of the grand coalition is negative with $-7/37$ because of the violation of the efficiency property. Therefore, the maximum surpluses can not be balanced. Hence, we need at least an additional iteration step to complete.

For the second iteration step, we use the vector $\mathbf{y}_{1}=(128/37,98/37,91/37,60/37)^{\top}$ while applying the procedure from above to get matrix $\mathbf{E}$ by
\begin{equation*}
\mathbf{E}=
\begin{bmatrix}[r]
  -1 & -1 & -1 & -1 &  0 &  0 &  1 \\
   1 &  0 & -1 & -1 & -1 & -1 &  1 \\
   1 &  1 & -1 &  1 & -1 & -1 &  1 \\
   0 & -1 &  1 & -1 &  1 &  1 &  1  
\end{bmatrix}.
\end{equation*}
With $\vec{\alpha}$ given by $(3,-3,-6,-6,-3,-3,10)^{\top}$. Quantifying now matrix $\mathbf{Q}$ through

\begin{equation*}
\mathbf{Q}=
\begin{bmatrix}[r]
   5 &  2 & -1 &  2 \\
   2 &  6 &  4 & -1 \\
  -1 &  4 &  7 & -4 \\
   2 & -1 & -4 &  6  
\end{bmatrix}.
\end{equation*}
The column vector $\mathbf{a}$ is given by $(22,31,16,7)^{\top}$. Solving this system of linear equations, $\mathbf{Q}\;\mathbf{x} - \mathbf{a} = \mathbf{0}$, we get as a solution $\mathbf{y}_{2}=(329/127,423/127,255/127,279/127)^{\top}$. The corresponding excess vector is given through 
\begin{equation*}
  \begin{split}
    exc(\mathbf{y}_{2}) & =(0, -329/127,-423/127,-371/127,-255/127,-584/127,-297/127,-245/127,\\ & -279/127, -608/127,-702/127,-269/127,-534/127,-863/127,-576/127,-16/127).
  \end{split}
\end{equation*}

We can check out that the maximum surpluses are still not balanced. Hence, we need an additional step to finally converge to the pre-nucleolus. By doing so, we solve the system of linear equations specified by 

\begin{equation*}
  \mathbf{Q}\,\mathbf{x} = \mathbf{a} \Longrightarrow 
  \begin{bmatrix}[r]
   6 &  4 &  0 &  1 \\
   4 &  6 &  2 &  1 \\
   0 &  2 &  7 & -4 \\
   1 &  1 & -4 &  6    
  \end{bmatrix} \cdot
    \begin{bmatrix}
    x_{1} \\
    x_{2} \\
    x_{3} \\
    x_{4}
  \end{bmatrix} =
    \begin{bmatrix}
    31 \\
    37 \\
    13 \\
    10
  \end{bmatrix}
\end{equation*}
The solution vector of this system of linear equations is given by $\mathbf{y}_{3}=(5/2,7/2,2,2)^{\top}$. By the subsequent symmetric matrix below, we can easily check out that the maximum surpluses at $\mathbf{y}_{3}$ are balanced.
\begin{equation*}
\mathbf{MSurp}=
\begin{bmatrix}[r]

          0   &        -5/2    &     -2     &    -2 \\
         -5/2 &         0      &     -2     &    -2 \\
         -2   &         -2     &     0      &    -2 \\
         -2   &         -2     &     -2     &     0  
\end{bmatrix}.
\end{equation*}

Hence, we have found the sole pre-kernel element of the game, which is also the nucleolus, because of zero-monotonicity of convex games.

Notice that in this specific case, we needed only three iteration steps to complete, which is below the theoretical-expected upper bound of iteration steps since by Theorem 9.2.1 of~\citet[p.~222]{mei:13}, we have $\binom{4}{2}-1 = 6-1=5$ to expect. Notice, this method is applicable for any $n$-person TU game (cf.~\citet[Appendix A]{mei:13}) and has also been proven to be useful in finding an N-shaped pre-kernel (cf.~\citet{mei:14b}). \hfill$\Diamond$
\end{example}

\section{Summary Assessment of the Algorithms}
\label{sec:concrem}

The conducted assessment has established that the new combinatorial algorithm of~\citet{maggiorano:2025} does not offer an efficient method to compute the nucleolus of convex games. This is owed to the fact that the authors have incorrectly applied the RGP so that the algorithm cannot lead to a correct computation of the nucleolus. Instead of providing a combinatorial and strongly polynomial algorithm, offering an alternative to the less practical ellipsoid method used by~\citet{faiglekernkuip:01}, their approach is at best suitable in testing the correctness of a presumed nucleolus solution. But even this needs further crucial corrections on the proposed algorithm. Despite its possible usefulness in the verification of a computed nucleolus solution, their suggested method possesses the intrinsic defect that it relies on the least-core to catch the nucleolus of a convex game. This makes it necessary to impose a submodular function minimization approach in testing the feasibility of the least-core constraints in polynomial-time. Thereby resulting in a relatively high order of $\tilde{O}(n^{8})$ in its runtime complexity even when one of the best-known general SFM algorithms is taken as a basis. All this makes their proposed method impractical even when considering only testing purposes on today's double-precision computers.  

Though conceived as a highly tailored algorithm specifically for the pre-kernel computation, the Fenchel-Moreau-based algorithm invented by~\citet{mei:13} is also applicable for the pre-nucleolus computation for the specific game  classes with a single-valued pre-kernel. That is to say, when the pre-kernel and pre-nucleolus coincide, then this approach can be used to find this unique point. For this specific case, \citeauthor{mei:13}'s method with a runtime complexity of $O(n^3)$ is a more efficient approach than the general combinatorial testing algorithm for the nucleolus on convex games as proposed by~\citet{maggiorano:2025}. Unfortunately, although the Fenchel-Moreau-based method, while effective for all game classes with a single-valued pre-kernel, might not have been recognized as a general-purpose, robust algorithm for the pre-nucleolus in the broader theoretical literature.

This is incomprehensible, since, as shown by~\citet{mei:13,mei:17f,mei:23,mei:18c}, this approach offers many more benefits, such as providing a simple, tractable computation method for the pre-kernel and pre-nucleolus. Besides a dual pre-kernel characterization, it offers a vast theoretical advantage as it eases the analysis of the pre-kernel and related solutions considerably. For this, the study by~\citet{mei:23} was revisited using an example from the literature in this regard. The focus is on how variations in game parameters, specifically coalition values, affect the pre-kernel solution, and in particular, the pre-nucleolus of a convex game. It is demonstrated for this example that under the non-empty interior condition, an exposed pre-kernel element from a default game can be maintained as a pre-kernel solution in a related game, even when coalition values are altered. This replication result is achieved by identifying a null space within the game space, allowing for variations in coalition values without disrupting the pre-kernel properties of the solution. This invokes the following implications for the pre-nucleolus:

\begin{enumerate}
\item {\bfseries Replication of a Pre-Kernel Point}: If the pre-kernel of a TU game consists of a single point and, in addition, it belongs to a payoff equivalence class with a non-empty interior, then it is replicable as a pre-kernel solution in a related game. Hence, the pre-nucleolus of the default game was replicated.
\item {\bfseries Stability of Pre-Kernel Properties}: The replication allows for variations in game parameters (like coalitional values) without altering the fundamental pre-kernel properties of the selected solution, effectively supporting the pre-imputation as a pre-kernel point. Hence, the default pre-nucleolus of the original game is a stable bargaining outcome.
\item {\bfseries Continuity of the Pre-Kernel Correspondence}: It is established that the pre-kernel correspondence (the mapping from games to pre-kernel points) can be single-valued and continuous on restricted subsets of the game space. Hence, the pre-nucleolus solution is stable on the convex domain of related games.   
\end{enumerate}

Finally, for the example, it was demonstrated how this pre-nucleolus can be manually computed by a step-by-step procedure using the underlying structure of the Fenchel-Moreau conjugation approach.

\footnotesize
\bibliography{polytime_algo_v2}

\end{document}